\documentclass[onecolumn,11pt,letter]{article}
\usepackage{epsfig}
\usepackage{epstopdf}
\usepackage{graphicx}
\usepackage{amsfonts}
\usepackage{amsmath,xspace}
\usepackage{amsthm}
\usepackage{proof}
\usepackage{latexsym}
\usepackage{amssymb}
\usepackage{color}
\usepackage{comment}
\usepackage[noend]{algorithmic}
\usepackage{algorithm}
\usepackage[margin=1in]{geometry}

\newtheorem{theorem}{Theorem}[section]
\newtheorem{lemma}[theorem]{Lemma}

\newtheorem{corollary}[theorem]{Corollary}

\newtheorem{definition}[theorem]{Definition}

\newtheorem{remark}[theorem]{Remark}
\DeclareMathAlphabet{\mathcal}{OMS}{cmsy}{m}{n}
\usepackage{tweaklist}
\renewcommand{\enumhook}{\setlength{\itemsep}{-0.5ex}}
\renewcommand{\itemhook}{\setlength{\itemsep}{-0.5ex}}

\newcommand{\executeiffilenewer}[3]{%
\ifnum\pdfstrcmp{\pdffilemoddate{#1}}%
{\pdffilemoddate{#2}}>0%
{\immediate\write18{#3}}\fi%
}

\newcommand{\torso}{\texttt{torso}}
\newcommand{\randset}{\textsc{RandomSet}}
\newcommand{\X}{{\mathcal{X}}}
\newcommand{\IS}{{\mathcal{I}_p}}

\newcommand{\dmway}{\textsc{Directed Multiway Cut}\xspace}
\newcommand{\dmcut}{\textsc{Directed Multicut}\xspace}

\DeclareMathOperator{\operatorClassNP}{NP}
\newcommand{\classNP}{\ensuremath{\operatorClassNP}}

\DeclareMathOperator{\operatorClassFPT}{FPT}
\newcommand{\classFPT}{\ensuremath{\operatorClassFPT}}
\DeclareMathOperator{\operatorClassW}{W}
\newcommand{\classW}[1]{\ensuremath{\operatorClassW[#1]}}

\newcommand*\samethanks[1][\value{footnote}]{\footnotemark[#1]}

\begin{document}

\title{Fixed-Parameter Tractability of Directed Multiway Cut Parameterized by the Size of the Cutset
\thanks{A preliminary version of this paper appeared in the proceedings of SODA 2012 \cite{DBLP:conf/soda/ChitnisHM12}.}}


\author{Rajesh Chitnis\thanks{Department of Computer Science , University of Maryland at
College Park, USA. Supported in part by a Google Faculty Research Award, an ONR Young Investigator Award, a NSF CAREER Award
and a DARPA BAA grant.  Email: \texttt{\{rchitnis, hajiagha\}@cs.umd.edu}} \and MohammadTaghi Hajiaghayi\samethanks[1] \and
D{\'a}niel Marx\thanks{Computer and Automation Research Institute, Hungarian Academy of Sciences (MTA SZTAKI), Budapest,
Hungary. Research  supported by the European Research Council (ERC)  grant 280152. Email: \texttt{dmarx@cs.bme.hu}} }


\maketitle

\begin{abstract}
  Given a directed graph $G$, a set of $k$ terminals and an integer
  $p$, the \textsc{Directed Vertex Multiway Cut} problem asks if there
  is a set $S$ of at most $p$ (nonterminal) vertices whose removal
  disconnects each terminal from all other terminals.
  \textsc{Directed Edge Multiway Cut} is the analogous problem where
  $S$ is a set of at most $p$ edges. These two problems indeed are
  known to be equivalent. A natural generalization of the multiway cut
  is the \textsc{Multicut} problem, in which we want to disconnect
  only a set of $k$ given pairs instead of all pairs. Marx
  (Theor. Comp. Sci. 2006) showed that in undirected graphs
  \textsc{Vertex/Edge Multiway} cut is fixed-parameter tractable (FPT)
  parameterized by $p$. Marx and Razgon (STOC 2011) showed that
  undirected \textsc{Multicut} is \classFPT\ and \dmcut is
  \classW1-hard parameterized by $p$. We complete the picture here by
  our main result which is that both \textsc{Directed Vertex Multiway
    Cut} and \textsc{Directed Edge Multiway Cut} can be solved in time
  $2^{2^{O(p)}}n^{O(1)}$, i.e., FPT parameterized by size $p$ of the
  cutset of the solution. This answers an open question raised by Marx
  (Theor. Comp. Sci. 2006) and Marx and Razgon (STOC 2011). It follows
  from our result that \textsc{Directed Edge/Vertex Multicut} is FPT
  for the case of $k=2$ terminal pairs, which answers another open
  problem raised in Marx and Razgon (STOC 2011).
\end{abstract}

%
%

\section{Introduction}
\label{section-introduction} Ford and Fulkerson~\cite{ford-fulkerson} gave the classical result on finding a minimum cut that
separates two terminals $s$ and $t$ in 1956.
A natural and well-studied generalization of the minimum $s-t$ cut problem is \textsc{Multiway Cut}, in which given a graph
$G$ and a set of terminals $\{s_1,s_2,\ldots,s_k\}$, the task is to find a minimum subset of vertices or edges whose
deletion disconnects all the terminals from one another. Dahlhaus et al.~\cite{johnson} showed the edge version in undirected
graphs is APX-complete for $k\geq 3$. For the edge version Karger et al.~\cite{karger} gave the current best known
approximation ratio of 1.3438 for general $k$. The vertex version of the problem is known to be at least as hard as the edge
version, and the current best approximation ratio is $2-\frac{2}{k}$ \cite{garg}.

The problem behaves very differently on directed graphs. Interestingly, for directed graphs, the edge and vertex versions turn
out to be equivalent.
Garg et al.~\cite{garg} showed that computing a minimum multiway cut in directed graphs is NP-hard and MAX SNP-hard already
for $k=2$. They also give an approximation algorithm with ratio $2$ log $k$, which was improved to ratio 2 later by Naor and
Zosin~\cite{naor}.

Rather than finding approximate solutions in polynomial time, one can look for exact solutions in time that is
superpolynomial, but still better than the running time obtained by brute force solutions. For example, Dahlhaus et
al.~\cite{johnson} showed that undirected \textsc{Multiway Cut} can be solved in time $n^{O(k)}$ on planar graphs, which can
be an efficient solution if the number of terminals is small. On the other hand, on general graphs the problem becomes NP-hard
already for $k=3$. In both the directed and the undirected version, brute force can be used to check in time $n^{O(p)}$ if a
solution of size at most $p$ exists: one can go through all sets of size at most $p$. Thus the problem can be solved in
polynomial time if the optimum is assumed to be small. In the undirected case, significantly better running time can be
obtained: the current fastest algorithms run in $O^*(2^p)$ time for both the vertex version~\cite{cygan-ipec-11} and the edge
version~\cite{DBLP:journals/mst/Xiao10} (the $O^{*}$ notation hides all factors which are polynomial in size of input). That
is, undirected \textsc{Multiway Cut} is fixed-parameter tractable parameterized by the size of the cutset we remove. Recall
that a problem is \emph{fixed-parameter tractable} (FPT) with a particular parameter $p$ if it can be solved in time
$f(p)n^{O(1)}$, where $f$ is an arbitrary function depending only on $p$; see \cite{downey-fellows,flum-grohe,niedermeier} for
more background. We give a brief summary of the race for faster FPT algorithms for \textsc{Undirected Multiway Cut} in
Figure~\ref{table:summary}.

\begin{figure}
\begin{center}
    \begin{tabular}{ | l | l | l |}
    \hline
    \text{Problem} & \text{Running Time} & \text{Paper} \\ \hline
    \text{Vertex Version} & Non-constructive FPT & Roberston and Seymour~\cite{DBLP:journals/jct/RobertsonS95b,DBLP:journals/jct/RobertsonS04}\\ \hline
    & $O^{*}(4^{p^3})$ & Marx~\cite{marx-2006} \\ \hline
    & $O^{*}(4^{p})$ & Chen et al.~\cite{chen-improved-multiway-cut}\\\hline
    & $O^{*}(4^{p})$ & Guillemot~\cite{DBLP:journals/disopt/Guillemot11a}\\\hline
    & $O^{*}(2^{p})$ & Cygan et al.~\cite{cygan-ipec-11}\\ \hline
    Edge Version & $O^{*}(2^{p})$ & Xiao~\cite{DBLP:journals/mst/Xiao10}\\ \hline
    \end{tabular}
\end{center}
\caption{Summary of FPT Results for \textsc{Undirected Multiway Cut}. Note that the $O^*$ notation hides all factors which are polynomial in the size of the input.\label{table:summary}}
\end{figure}

Our main result is that the directed version of \textsc{Multiway Cut} is also fixed-parameter tractable:

\begin{theorem}{\bf (main result)}
  \label{thm-main} \textsc{Directed Vertex Multiway Cut} and
  \textsc{Directed Edge Multiway Cut} can be solved in
  $O^{*}(2^{2^{O(p)}})$ time.
\end{theorem}

Note that the hardness result of Garg et al.~\cite{garg} shows that in the directed case the problem is nontrivial (in fact,
NP-hard) even for $k=2$ terminals; our result holds without any bound on the number of terminals.  The question was first
asked explicitly in \cite{marx-2006} and was also stated as an open problem in \cite{marx-stoc-2011}. Our result shows in
particular that directed multiway cut is solvable in polynomial time if the size of the optimum solution is $O(\log \log n)$,
where $n$ is the number of vertices in the digraph.

A more general problem is \textsc{Multicut}: the input contains a set $\{(s_1,t_1),\dots,(s_k,t_k)\}$ of $k$ pairs, and the
task is to break every path from $s_i$ to its corresponding $t_i$ by the removal of at most $p$ vertices. Very recently, it
was shown that undirected \textsc{Multicut} is FPT parameterized by $p$
\cite{DBLP:journals/corr/abs-1010-5197,marx-stoc-2011}, but the directed version is unlikely to be FPT as it is W[1]-hard
\cite{marx-stoc-2011} with this parameterization.  However, in the special case of $k=2$ terminal pairs, there is a simple
reduction from \dmcut to \dmway, thus our result shows that the latter problem is FPT parameterized by $p$ for $k=2$.  Let us
briefly sketch the reduction. (Note that the reduction we sketch works only for the variant of \dmcut which allows the
deletion of terminals. Marx and Razgon~\cite{marx-stoc-2011} asked about the FPT status of this variant which is in fact
equivalent to the one which does not allow deletion of the terminals.) Let $(G,T,p)$ be a given instance of \dmcut and let
$T=\{(s_1,t_1),(s_2,t_2)\}$. We construct an equivalent instance  of \dmway as follows: Graph $G'$ is obtained by adding two
new vertices $s,t$ to the graph and adding the four edges $s\rightarrow s_1$, $t_1\rightarrow t$, $t\rightarrow s_2$, and
$t_2\rightarrow s$. It is easy to see that the \dmway instance $(G',\{s,t\},p)$ is equivalent to the original \dmcut
instance.\footnote{$G$ has a $s_i\rightarrow t_i$ path for some $i$ if and only if $G'$ has a $s\rightarrow t$ or
$t\rightarrow s$ path. This is because $G$ has a $s_1\rightarrow t_1$ path if and only if $G'$ has a $s\rightarrow t$ path and
$G$ has a $s_2\rightarrow t_2$ path if and only if $G'$ has a $t\rightarrow s$ path. This property of paths also holds after
removing some vertices/edges and thus the two instances are equivalent.
}


\begin{corollary}\label{cor:dirmulticut}
\dmcut with $k=2$ can be solved in time $O^{*}(2^{2^{O(p)}})$.
\end{corollary}

The complexity of the case $k=3$ remains an interesting open problem.
\medskip

{\bf Our techniques.} Our algorithm for \dmway is inspired by the algorithm of Marx and Razgon~\cite{marx-stoc-2011} for
undirected \textsc{Multicut}. In particular we use the technique of ``random sampling of important separators'' introduced in
\cite{marx-stoc-2011} and try to ensure that there is a solution whose ``isolated part'' is empty. However, \dmway behaves in
a significantly different way than \textsc{Multicut}: at the same time, we are dealing with a much easier and a much harder
situation. The first step in \cite{marx-stoc-2011} is to reformulate the problem in a way that the solution has to be a
multiway cut of a certain set $W$ of vertices; the technique of {\em iterative compression} allows us to reduce the original
problem to this new version. As \textsc{Multiway Cut} is already defined in terms of finding a multiway cut, this step is not
necessary in our case. Furthermore, in \cite{marx-stoc-2011}, after ensuring that there is a solution whose ``isolated part''
is empty, the problem is reduced to \textsc{Almost-2SAT} (Given a 2SAT formula and an integer $k$, is there an assignment
satisfying all but $k$ of the clauses ?) This reduction works only if every component has at most two ``legs''; a delicate
branching algorithm is given to ensure this property. In the case of \dmway, the situation is much simpler: if there is a
solution whose ``isolated part'' is empty, then the problem can be reduced to the undirected version and then we can use the
current fastest undirected algorithm ~\cite{cygan-ipec-11}, which runs in $O^{*}(2^p)$ time.

On the other hand, the fact that we are dealing with a directed graph makes the problem significantly harder (recall that
\dmcut is W[1]-hard parameterized by $p$, thus it is expected that not every undirected argument generalizes to the directed
case). After defining a proper notion of directed important separators, the non-trivial interaction amongst two kinds of
``shadows'' forces us to do the random sampling of important separators in two independent steps and the analysis becomes more delicate.

\medskip

{\bf Independent and followup work.} The fixed-parameter tractability of \textsc{Multicut} in undirected graphs parameterized
only by the size of the cutset was shown independently by Marx and Razgon~\cite{marx-stoc-2011} and Bousquet et
al.~\cite{DBLP:journals/corr/abs-1010-5197}. Marx and Razgon~\cite{marx-stoc-2011} also showed that \dmcut is W[1]-hard
parameterized by the size of the cutset. The technique of random sampling of important separators introduced
in~\cite{marx-stoc-2011} is a crucial element of our algorithm. A very different application of this technique was given by
Lokshtanov and Marx~\cite{lokshtanov-marx-clustering} in the context of clustering problems.

The preliminary version of this paper adapted the framework of random
sampling of important separators to directed graphs and showed the
fixed-parameter tractability of \dmway parameterized by the size of
the cutset. This framework was later used by Kratsch et
al.~\cite{multicut-dags} to show the fixed-parameter tractability of
\dmcut on directed acyclic graphs and by Chitnis et
al.~\cite{DBLP:conf/icalp/ChitnisCHM12} to show the fixed-parameter
tractability of \textsc{Subset Directed Feedback Vertex Set}. The
latter paper improved the randomized sampling process to make the
algorithms more efficient; in particular, this improvement results in
a $O^*(2^{O(p^2)})$ algorithm for \dmway.  The question of
existence of a polynomial kernel for \dmway was answered negatively by
Cygan et al.~\cite{cygan-no-poly-kernel} who showed that \dmway (even
for two terminals) does not have a polynomial kernel unless \classNP$\
\subseteq\ $coNP/poly and the polynomial hierarchy collapses to the
third level. An interesting open question is the complexity of \dmcut
for $k=3$ or with combined parameters $k$ and $p$.

\section{Preliminaries}
\label{section-preliminaries}

A multiway cut is a set of edges/vertices that separate the terminal vertices from each other:
\begin{definition}{\bf (multiway cut)}
\label{defn-directed-multiway-cut} Let $G$ be a directed graph and let $T=\{t_1,t_2,\ldots,t_k\}\subseteq V(G)$ be a set of
terminals.
\begin{enumerate}
\item $S\subseteq V(G)$ is a \emph{vertex multiway cut} of $(G,T)$ if $G\setminus S$ does not have a path from $t_i$ to
    $t_j$ for any $i\neq j$.
\item $S\subseteq E(G)$ is a \emph{edge multiway cut} of $(G,T)$ if $G\setminus S$ does not have a path from $t_i$ to
    $t_j$ for any $i\neq j$.
\end{enumerate}
\end{definition}

\hspace{-6mm}In the edge case, it is straightforward to define the problem we are trying to solve:
\begin{center}
\noindent\framebox{\begin{minipage}{6.00in}
\textbf{\textsc{Directed Edge Multiway Cut}}\\
\emph{Input }: A directed graph $G$, an integer $p$ and a set of terminals $T$.\\
\emph{Output} : A multiway cut $S\subseteq E(G)$ of $(G,T)$ of size at most $p$ or ``NO'' if such a multiway cut does not
exist.
\end{minipage}}
\end{center}

\hspace{-6mm}In the vertex case, there is a slight technical issue in the definition of the problem: are the terminal vertices
allowed to be deleted?  We focus here on the version of the problem where the vertex multiway cut we are looking for has to be
disjoint from the set of terminals. More generally, we define the problem in such a way that the graph has some
\emph{distinguished} vertices which cannot be included as part of any separator (and we assume that every terminal is a
distinguished vertex). This can be modeled by considering weights on the vertices of the graph: weight of $\infty$ on each
distinguished vertex and 1 on every non-distinguished vertex. We only look for solutions of finite weight. From here on, for a
graph $G$ we will denote by $V^{\infty}(G)$ the set of distinguished vertices of $G$ with the meaning that these
distinguished vertices cannot be part of any separator, i.e., all separators we consider are of finite weight. In fact, for
any separator we can talk interchangeably about size or weight as these notions are the same since each vertex of separator
has weight 1.

The main focus of the paper is the following vertex version, where we require $T\subseteq V^{\infty}(G)$, i.e., terminals
cannot be deleted:

\begin{center}\noindent\framebox{\begin{minipage}{6.00in}
\textbf{\textsc{Directed Vertex Multiway Cut}}\\
\emph{Input }: A directed graph $G$, an integer $p$, a set of terminals $T$ and a set $V^{\infty}\supseteq T$ of distinguished vertices.\\
\emph{Output} : A multiway cut $S\subseteq V(G)\setminus V^{\infty}(G)$ of $(G,T)$ of size at most $p$ or ``NO'' if such a
multiway cut does not exist.
\end{minipage}}
\end{center}
We note that if we want to allow the deletion of the terminal vertices, then it is not difficult to reduce the problem to the
version defined above. For each terminal $t$ we introduce a new vertex $t'$ and we add the directed edges $(t,t')$ and
$(t',t)$. Let the new graph be $G'$ and let $T'=\{t'\ |\ t\in T\}$. Then there is a clear bijection between vertex multiway
cuts which can include terminals in the instance $(G,T,p)$ and vertex multiway cuts which cannot include terminals in the
instance $(G',T',p)$.

The two versions \textsc{Directed Vertex Multiway Cut} and
\textsc{Directed Edge Multiway Cut} defined above are known to be
equivalent. For sake of completeness, we prove the equivalence in
Section~\ref{sec:equiv-textscd-vert}. In the remaining part of the
paper, we concentrate on finding an FPT algorithm for \textsc{Directed
  Vertex Multiway Cut}, which we henceforth call \dmway for brevity.

\subsection{Equivalence of Vertex and Edge versions of \dmway}
\label{sec:equiv-textscd-vert}

We first show how to solve the vertex version using the edge
version. Let $(G,T,p)$ be a given instance of \textsc{Directed Vertex
  Multiway Cut} and let $V^{\infty}(G)$ be the set of distinguished
vertices. We construct an equivalent instance $(G',T',p)$ of
\textsc{Directed Edge Multiway Cut} as follows. Let the set $V'$
contain two vertices $v^{\text{in}}$, $v^{\text{out}}$ for every $v\in
V(G)\setminus V^{\infty}(G)$ and a single vertex
$u^{\text{in}}=u^{\text{out}}$ for every $u\in V^{\infty}(G)$. The
idea is that all incoming/outgoing edges of $v$ in $G$ will now be
incoming/outgoing edges of $v^{\text{in}}$ and $v^{\text{out}}$,
respectively. For every vertex $v\in V(G)\setminus V^{\infty}(G)$, add
an edge $(v^{\text{in}},v^{\text{out}})$ to $G'$. Let us call these as
Type I edges. For every edge $(x,y)\in E(G)$, add $(p+1)$ parallel
$(x^{\text{out}},y^{\text{in}})$ edges. Let us call these as Type II
edges. Define $T'=\{v^{\text{in}}\ |\ v\in T\}$. Note that the number
of terminals is preserved. We have the following lemma:
\begin{lemma}
  $(G,T,p)$ is a yes-instance
  of \label{lemma-solving-vertex-using-edge} \textsc{Directed Vertex
    Multiway Cut} if and only if $(G',T',p)$ is a yes-instance of
  \textsc{Directed Edge Multiway Cut}.
\end{lemma}
\begin{proof}
Suppose $G$ has a vertex multiway cut, say $S$, of size at most $p$. Then the set $S'=\{(v^{\text{in}},v^{\text{out}})\ |\ v\in S\}$ is clearly a edge
multiway cut for $G'$ and $|S'|=|S|\leq p$.

Suppose $G'$ has an edge multiway cut say $S'$ of size at most $p$. Note that it does not help to pick in $S$ any edges of
Type II as each edge has $(p+1)$ parallel copies and our budget is $p$. So let $S=\{v\ |\ (v^{\text{in}},v^{\text{out}})\in S'\}$. Then $S$ is a
vertex multiway cut for $G$ and $|S|\leq |S'|\leq p$.
\end{proof}

We now show how to solve the edge version using the vertex version. Let $(G,T,p)$ be a given instance of \textsc{Directed Edge
Multiway Cut}. We construct an equivalent instance $(G',T',p)$ of \textsc{Directed Vertex Multiway Cut} as follows. For each
vertex $u\in V(G)\setminus T$, create a set $C_u$ which contains $u$ along with $p$ other copies of $u$. For $t\in T$ we let
$C_t=\{t\}$. For each edge $(u,v)\in E(G)$ create a vertex $\beta_{uv}$. Add edges $(x,\beta_{uv})$ for all $x\in C_u$ and
$(\beta_{uv},y)$ for all $y\in C_v$. Define $T'=\bigcup_{t\in T} C_t = T$. Let $V^{\infty}(G') = T'$

\begin{lemma}
  $(G,T,p)$ is a yes-instance
  of \label{lemma-solving-edge-using-vertex} \textsc{Directed Edge
    Multiway Cut} if and only if $(G',T',p)$ is a yes-instance of
  \textsc{Directed Vertex Multiway Cut}.
\end{lemma}
\begin{proof}
Suppose $G$ has an edge multiway cut, say $S$, of size at most $p$. Then the set $S'=\{\beta_{uv}\ |\ (u,v)\in S\}$ is clearly a
vertex multiway cut for $G'$ and $|S'|=|S|\leq p$.

Suppose $G'$ has a vertex multiway cut say $S'$ of size at most $p$. Note that it does not help to pick in $S$ any vertices
from the $C_z$ of any vertex $z\in V(G)\setminus T$ as each vertex has $(p+1)$ equivalent copies and our budget is $p$. So let
$S=\{(u,v)\ |\ \beta_{uv}\in S'\}$. Then $S$ is a edge multiway cut for $G$ and $|S|\leq |S'|\leq p$.
\end{proof}


\subsection{Separators and Shadows}
\label{subsection:separators-and-shadows}

The crucial idea in the algorithm of \cite{marx-stoc-2011} for (the vertex version of) undirected \textsc{Multicut} is to get rid of the ``isolated
part'' of the solution $S$. We use a similar concept here, but we use the term {\em shadow,} as it is more expressive for
directed graphs.

\begin{definition}{\bf (separator)}\label{defn-sep}
Let $G$ be a directed graph and $V^{\infty}(G)\supseteq T$ be the set of distinguished (``undeletable'') vertices. Given two disjoint
non-empty sets $X,Y\subseteq V$  we call a set $S\subseteq V\setminus (X\cup Y\cup V^{\infty})$ an \emph{$X-Y$ separator} if
there is no path from $X$ to $Y$ in $G\setminus S$. A set $S$ is a {\em minimal $X-Y$ separator} if no proper subset of $S$ is
an $X-Y$ separator.
\end{definition}

Note that here we explicitly define the $X-Y$ separator $S$ to be disjoint from $X$ and $Y$.
\begin{definition}{\bf (shadows)}
\label{defn-f-and-r-and-shadow} Let $G$ be graph and $T$ be a set of terminals. Let $S\subseteq V(G)\setminus V^{\infty}(G)$
be a subset of vertices.
\begin{enumerate}
\item The \emph{forward shadow} $f_{G,T}(S)$ of $S$ (with respect to $T$) is the set of vertices $v$ such that $S$ is a
    $T-\{v\}$ separator in $G$.
\item The \emph{reverse shadow} $r_{G,T}(S)$ of $S$ (with respect to $T$) is the set of vertices $v$ such that $S$ is a
    $\{v\}-T$ separator in $G$.
\end{enumerate}
The \emph{shadow} of $S$ (with respect to $T$) is the union of $f_{G,T}(S)$ and $r_{G,T}(S)$.
\end{definition}
That is, we can imagine $T$ as a light source with light spreading on the directed edges. The forward shadow is the set of
vertices that remain dark if the set $S$ blocks the light, hiding $v$ from $T's$ sight. In the reverse shadow, we imagine that
light is spreading on the edges backwards. We abuse the notation slightly and write $v-T$ separator instead of $\{v\}-T$
separator. We also drop $G$ and $T$ from the subscript if they are clear from the context. Note that $S$ itself is not in the
shadow of $S$ (as, by definition, a $T-v$ or $v-T$ separator needs to be disjoint from $T$ and $v$), that is, $S$ and $f_{G,T}(S)\cup
r_{G,T}(S)$ are disjoint. See Figure~\ref{fig:shadow} for an illustration.

\begin{figure}[t]
\centering
\def\svgwidth{0.4\linewidth}%
\executeiffilenewer{mway1.svg}{mway1.pdf}%
{inkscape -z -D --file=mway1.svg %
--export-pdf=mway1.pdf --export-latex}%
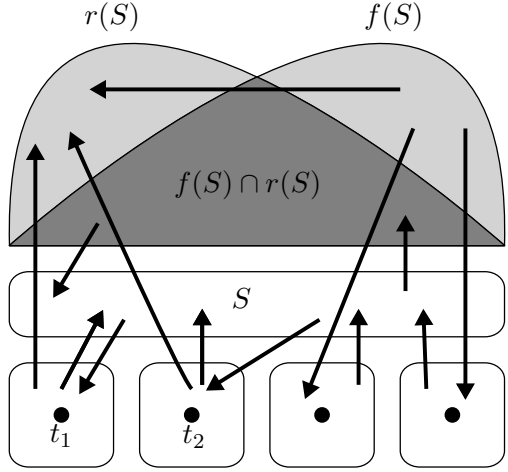%

\caption{For every vertex $v\in f(S)$, the set $S$ is a $T-v$
  separator. For every vertex $w\in r(S)$, the set $S$ is a $w-T$
  separator. For every vertex $y\in f(S)\cap r(S)$, the set $S$ is
  both a $T-y$ and $y-T$ separator. Finally for every $z\in
  V(G)\setminus [S\cup r(S)\cup f(S)\cup T]$, there are both $z-T$ and
  $T-z$ paths in the graph $G\setminus S$\label{fig:shadow}. Note that
  every such vertex $z$ belongs to a strongly connected component of
  $G\setminus S$ containing $T$ and there are no edges between these
  components.}
\end{figure}

\section{Overview of our Algorithm}
\label{overview}

We say that a solution $S$ of \dmway is {\em shadowless} (with respect
to $T$) if $f(S)=r(S)=\emptyset$. The following lemma shows the
importance of \emph{shadowless solutions} for \dmway. Clearly, any
solution of the underlying undirected instance (where we disregard the
orientation of the edges) is a solution for \dmway cut. The converse
is not true in general: a solution of the directed problem is a not
always solution of the undirected problem. However, the converse statement is true for shadowless
solutions of the directed instance:

\begin{lemma}
\label{lemma-redn-to-undirected-case} Let $G^{*}$ be the underlying undirected graph of $G$. If $S$ is a \emph{shadowless}
solution for an instance $(G,T,p)$ of \dmway, then $S$ is also a solution for the instance $(G^{*},T,p)$ of \textsc{Undirected
Multiway Cut}.
\end{lemma}
\begin{proof}
If $S$ is a shadowless solution, then for each vertex $v$ in $G\setminus S$, there is a $t_1\to v$ path and a $v \to t_2$ path
for some $t_1,t_2\in T$.  As $S$ is a solution, it is not possible that $t_1\neq t_2$: this would give a $t_1\rightarrow t_2$
path in $G\setminus S$. Therefore, if $S$ is a shadowless solution, then each vertex in the graph $G\setminus S$ belongs to
the strongly connected component of exactly one terminal.
A directed edge between the strongly connected components of $t_i$ and $t_j$ would imply the existence of either a $t_i\to
t_j$ or a $t_j\to t_i$ path, which contradicts the fact that $S$ is a solution of the \dmway instance. Hence the strongly
connected components of $G\setminus S$ are exactly the same as the weakly connected components of $G\setminus S$, i.e., $S$ is
also a solution for the underlying instance of \textsc{Undirected Multiway Cut}.
\end{proof}

\begin{figure}[t]
\centering
\def\svgwidth{0.4\linewidth}%
\executeiffilenewer{mway2.svg}{mway2.pdf}%
{inkscape -z -D --file=mway2.svg %
--export-pdf=mway2.pdf --export-latex}%
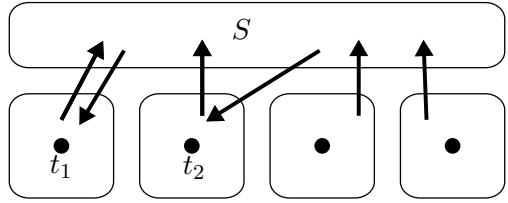%

\caption{A shadowless solution $S$ for a \dmway instance. Every vertex
  of $G\setminus S$ is in the strongly connected component of some
  terminal $t_i$.  There are no edges between the strongly connected
  components of the terminals $t_i$, thus $S$ is also a solution of
  the underlying \textsc{Undirected Multiway Cut}
  instance. \label{fig:constant}}
\end{figure}

An illustration of Lemma~\ref{lemma-redn-to-undirected-case} is given in Figure~\ref{fig:constant}.
Lemma~\ref{lemma-redn-to-undirected-case} shows that if we can transform the instance in a way that ensures the existence of a
shadowless solution, then we can reduce the problem to undirected \textsc{Multiway Cut} and use the $O^*(4^p)$ algorithm for
that problem due to Guillemot~\cite{DBLP:journals/disopt/Guillemot11a} which can handle the case when there are some
distinguished vertices similar to what we consider.
Our transformation is based on two ingredients: random sampling of important separators and
reduction of the instance using the \torso\ operation. These techniques were introduced by Marx and
Razgon~\cite{marx-stoc-2011} for the undirected \textsc{Multicut} problem. In Section~\ref{subsection-imp-separators}, we
review these tools and adapt them for directed graphs.  \medskip

\textbf{Random sampling of important separators.} As a first step of reducing the problem to a shadowless instance, we need a set $Z$
that has the following property:

\begin{quote}
There is a solution $S^*$ such that $Z$ contains the shadow of $S^*$, but $Z$ is disjoint from $S^*$. \hfill (*)
\end{quote}
If we have a set $Z$ that satisfies Property (*), we modify the
instance in a way that removes the set $Z$. The modification is done
such that $S^*$ remains a solution of the reduced instance; in fact,
it becomes a shadowless solution. This means
that the problem can be solved by
Lemma~\ref{lemma-redn-to-undirected-case}. This process of getting rid
of the set $Z$ in an appropriate way is accomplished by the \torso\
operation defined below.

Unfortunately, when we are
trying to construct the set $Z$, we do not know anything about the
solutions of the instance and in particular we have no way of checking
if a given set $Z$ satisfies Property (*). Nevertheless, we use a
randomized procedure that creates a set $Z$ and we give a lower bound
on the probability that $Z$ satisfies Property (*). For the
construction of this set $Z$, we use a very specific probability
distribution that was introduced in \cite{marx-stoc-2011}. This
probability distribution is based on randomly selecting ``important
separators'' and taking the union of their shadows. At this point, we
can consider the sampling as a black-box function
``\randset$(G,T,p)$'' that returns a random subset $Z\subseteq V(G)$
according to a probability distribution that satisfies certain
properties. The precise description of this function and the
properties of the distribution it creates is described in
Section~\ref{sec:random-sampling-1} (see
Theorem~\ref{th:random-sampling}). The randomized selection can be
derandomized: the randomized selection can be turned into a
deterministic algorithm that returns a bounded number of sets such
that at least one of them satisfies the required property
(Section~\ref{subsection-derandomization}). To make the description of
the algorithm simpler, we focus on the randomized algorithm in this
section.

\textbf{Torsos.} We use the function $\randset(G,T,p)$ to construct a set $Z$ of vertices that we want to get rid of. However
we must be careful: when getting rid of the set $Z$ we should ensure that the information relevant to $Z$ is captured in the
reduced instance. This is exactly accomplished by the \torso\ operation which removes a set of vertices without making the
problem any easier. We formally define this operation as follows:
\medskip

\begin{definition}{\bf (torso)}
\label{defn-torso} Let $G$ be a directed graph and let $C\subseteq V(G)$. The graph \emph{\torso}$(G,C)$ has vertex set $C$
and there is a (directed) edge $(a,b)$ in \emph{\torso}$(G,C)$ if there is an $a\rightarrow b$ path in $G$ whose internal
vertices are not in $C$.
\end{definition}

\begin{figure}[t]
\centering
\includegraphics[width=7in]{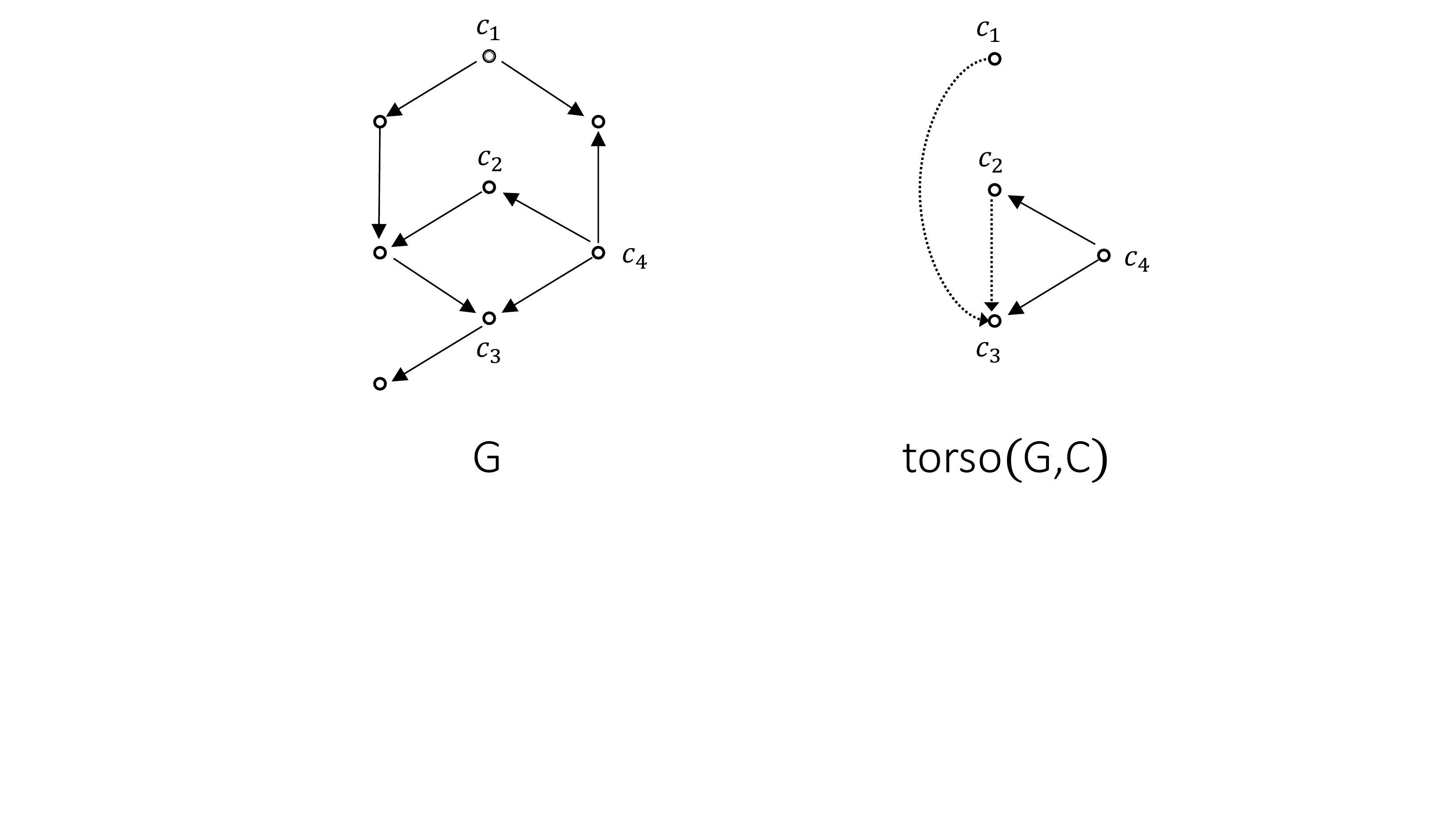}
\vspace{-45mm}
\caption{Let $C=\{c_1,c_2,c_3,c_4\}$. In the graph \torso$(G,C)$ the edges $(c_4,c_3)$ and $(c_4,c_2)$ carry over from $G$. The new edges (shown by dotted arrows) that get added because of the \torso\ operation are
$(c_1,c_3)$ and $(c_2,c_3)$.\label{fig:torso}}
\end{figure}

See Figure~\ref{fig:torso} for an example of the \torso\ operation. Note that if $a,b\in C$ and $(a,b)$ is a directed edge of
$G$, then $\torso(G,C)$ contains $(a,b)$ as well. Thus $G[C]$, which is the graph induced by $C$ in $G$, is a subgraph of
$\torso(G,C)$. The following lemma shows that the \torso\ operation
preserves separation inside $C$.

\begin{lemma}
{\bf (torso preserves separation)} \label{torso-preserves-separation} Let $G$ be a directed graph and $C\subseteq V(G)$. Let
$G'=\emph{\torso}(G,C)$ and $S\subseteq C$. For $a,b \in C\setminus S$, the graph $G\setminus S$ has an $a\rightarrow b$ path if
and only if $G'\setminus S$ has an $a\rightarrow b$ path.
\end{lemma}
\begin{proof}
Let $P$ be a path from $a$ to $b$ in $G$. Suppose $P$ is disjoint from $S$. Then $P$ contains vertices from $C$ and
$V(G)\setminus C$. Let $u,v$ be two vertices of $C$ such that  every vertex of $P$ between $u$ and $v$ is from $V(G)\setminus
C$. Then by definition there is an edge $(u,v)$ in $\torso(G,C)$. Using such edges we can modify $P$ to obtain an $a\rightarrow
b$ path that lies completely in $\torso(G,C)$ but avoids $S$.

Conversely suppose $P'$ is an $a\rightarrow b$ path in $\torso(G,C)$ and it avoids $S\subseteq C$. If $P'$ uses an edge
$(u,v)\notin E(G)$, then this means that there is a $u\rightarrow v$ path $P''$ whose internal vertices are not in $C$. Using
such paths we modify $P$ to get an $a\rightarrow b$ path $P_0$ that only uses edges from $G$. Since $S\subseteq C$ we have
that the new vertices on the path are not in $S$ and so $P_0$ avoids $S$.
\end{proof}

If we want to remove a set $Z$ of vertices, then we create a new instance by taking the \torso\ on the {\em complement} of
$Z$:
\begin{definition}
\label{defn-reduced-instance} Let $I=(G,T,p)$ be an instance of \dmway and $Z\subseteq V(G)\setminus T$. The reduced instance
$I/Z=(G',T',p)$ is defined as
\begin{itemize}
    \item $G'=$\emph{\torso}$(G,V(G)\setminus Z)$
    \item $T'=T$
\end{itemize}
\end{definition}
The following lemma states that the operation of taking the \torso\ does not make the \dmway problem easier for any
$Z\subseteq V(G)\setminus T$ in the sense that any solution of the reduced instance $I/Z$ is a solution of the original
instance $I$. Moreover, if we perform the \torso\ operation for a $Z$ that is large enough to contain the shadow of some
solution $S^*$ but at the same time small enough to be disjoint from $S^*$, then $S^*$ remains a solution for the reduced
instance $I/Z$ and in fact it is a shadowless solution for $I/Z$. Therefore, our goal is to randomly select a set $Z$ in a way
that we can bound the probability that $Z$ satisfies Property (*) defined above for some hypothetical solution $S^*$.

\begin{lemma}
{\bf (creating a shadowless instance)} \label{reduced-instance-soln} Let $I=(G,T,p)$ be an instance of \dmway and $Z\subseteq
V(G)\setminus T$.
\begin{enumerate}
\item If $S$ is a solution for $I/Z$, then $S$ is also a solution for $I$.
\item If $S$ is a solution for $I$ such that $f_{G,T}(S)\cup r_{G,T}(S)\subseteq Z$ and $S\cap Z=\emptyset$, then $S$ is a
    shadowless solution for $I/Z$.
\end{enumerate}
\end{lemma}
\begin{proof}
  Let $G'$ be the graph $\torso(G,V(G)\setminus Z)$. To prove the
  first part, suppose that $S\subseteq V(G')$ is a solution for $I/Z$
  and $S$ is not a solution for $I$. Then there are terminals
  $t_1,t_2\in T$ such that there is an $t_1\rightarrow t_2$ path $P$
  in $G\setminus S$. As $t_1,t_2\in T$ and $Z\subseteq V(G)\setminus
  T$, we have that $t_1,t_2\in V(G)\setminus Z$. In fact, we have
  $t_1,t_2\in (V(G)\setminus Z)\setminus S$. Lemma
  \ref{torso-preserves-separation} implies that there is an
  $t_1\rightarrow t_2$ path in $G'\setminus S$, which is a
  contradiction as $S$ is a solution for $I/Z$.

For the second part of the lemma, let $S$ be a solution for $I$ such that $S\cap Z = \emptyset$ and $f_{G,T}(S)\cup
r_{G,T}(S)\subseteq Z$.
We want to show that $S$ is a shadowless solution for $I/Z$. First we show that $S$ is a solution for $I/Z$. Suppose to the
contrary that there are terminals $x',y'\in T'(=T)$ such that $G'\setminus S$ has an $x'\rightarrow y'$ path. As $x',y'\in V(G)\setminus
Z$, Lemma \ref{torso-preserves-separation} implies $G\setminus S$ also has an $x'\rightarrow y'$ path, which is a
contradiction as $S$ is a solution of $I$.

Finally, we show that $S$ is shadowless in $I/Z$, i.e., $r_{G',T}(S)=\emptyset=f_{G',T}(S)$. We only prove that
$r_{G',T}(S)=\emptyset$: the argument for $f_{G',T}(S)=\emptyset$ is analogous. Assume to the contrary that there exists $w\in
r_{G',T}(S)$ (note that we have $w\in V(G')$, i.e., $w\notin Z$). So $S$ is a $w-T$ separator in $G'$, i.e., there is no $w-T$
path in $G'\setminus S$. Lemma~\ref{torso-preserves-separation} gives that there is no $w-T$ path in $G\setminus S$, i.e.,
$w\in r_{G,T}(S)$. But $r_{G,T}(S)\subseteq Z$ and so we have $w\in Z$ which is a contradiction. Thus $r_{G,T}(S)\subseteq Z$
in $G$ implies that $r_{G',T}(S)=\emptyset$.
\end{proof}

\begin{algorithm}[t]
\caption{\textsc{FPT Algorithm for \dmway}} \label{alg:algo}
~\\
\textbf{Input}: An instance $I_1=(G_1,T,p)$ of \dmway.\\

\begin{algorithmic}[1]

    \STATE Let $Z_1=\randset(G_1,T,p)$.
        \STATE Let $G_2 = (G_1)_{\textup{rev}}$\hspace{63mm}\COMMENT {Reverse the orientation of every edge}
        \STATE Let $V^{\infty}(G_2)=V^{\infty}(G_1)\cup Z_1$.\hspace{42mm}\COMMENT {Set weight of every vertex of $Z_1$ to $\infty$}
        \STATE Let $Z_2=\randset(G_2,T,p)$.
        \STATE Let $Z=Z_1\cup Z_2$.
        \STATE Let $G_3=\torso(G_1,V(G)\setminus Z)$.\hspace{81mm}\COMMENT {Get rid of $Z$}
        \STATE Solve the underlying \emph{undirected} instance $(G^*_3,T,p)$ of \textsc{Multiway Cut}.
        \IF {$(G^*_3,T,p)$ has a solution $S$}
                 \STATE \textbf{return} $S$
        \ELSE
                 \STATE \textbf{return} ``NO"
        \ENDIF
\end{algorithmic}
\end{algorithm}

\textbf{The Algorithm.} The description of our algorithm is given in Algorithm~\ref{alg:algo}. Recall that we are trying to
solve a version of \dmway where we are given a set $V^{\infty}$ of distinguished vertices which are undeletable, i.e., have
infinite weight.

Due to the delicate way separators behave in directed graphs, we construct the set $Z$ in two phases, calling the function
$\randset$ twice. Our aim is to show that there is a solution $S$ such that we can give a lower bound on the probability that
$Z_1$ contains $r_{G_1,T}(S)$ and $Z_2$ contains $f_{G_1,T}(S)$. Note that the graph $G_2$ obtained in Step 2 depends on the
set $Z_1$ returned in Step 1 (as we made the weight of every vertex in $Z_1$ infinite), thus the distribution of the second
random sampling depends on the result $Z_1$ of the first random sampling.  This means that we cannot make the two calls in
parallel.

We use the \torso\ operation to remove the vertices in $Z=Z_1\cup Z_2$ (Step 5), and then solve the undirected
\textsc{Multiway Cut} instance obtained by disregarding the orientation of the edges. For this purpose, we can use the
algorithm of Guillemot~\cite{DBLP:journals/disopt/Guillemot11a} that solves the undirected problem in time $O^*(4^p)$. Note
that the algorithm for undirected \textsc{Multiway Cut} in~\cite{DBLP:journals/disopt/Guillemot11a} explicitly considers the
variant where we have a set of distinguished vertices which cannot be deleted.

The following two lemmas show that Algorithm~\ref{alg:algo} is a correct randomized algorithm. One direction is easy to see: the algorithm has no false positives.
\begin{lemma}
\label{lem:correctalg-part-one} Let $I_1=(G_1,T,p)$ be an instance of \dmway.
If Algorithm~\ref{alg:algo} returns a set $S$, then $S$ is a solution for $I_1$.
\end{lemma}
\begin{proof}
Any solution $S$ of the undirected instance $(G^*_3,T,p)$ returned by Algorithm~\ref{alg:algo} is clearly a solution of the
directed instance $(G_3,T,p)$ as well. By Lemma~\ref{reduced-instance-soln}(1) the \torso\ operation does not make the problem
easier by creating new solutions. Hence $S$ is also a solution for $I_1 = (G_1,T,p)$
\end{proof}

The following lemma shows that if the instance has a solution, then
the algorithm finds one with certain probability.
\begin{lemma}
\label{lem:correctalg} Let $I_1=(G_1,T,p)$ be an instance of \dmway.
If $I_1$ is a yes-instance of \dmway, then Algorithm~\ref{alg:algo} returns a set $S$ which is a solution for $I$ with
probability at least $2^{-2^{O(p)}}$.
\end{lemma}

By Lemma~\ref{reduced-instance-soln}(2), we can prove Lemma~\ref{lem:correctalg} by showing that if $I_1$ is a yes-instance,
then there exists a solution $S^*$ such that $Z$ satisfies the two requirements $Z\cap S=\emptyset$ and $f_{G_1,T}(S)\cup
r_{G_1,T}(S)\subseteq Z$ with suitable probability. This requires a deeper analysis of the structure of optimum solutions and
the probability distribution behind the function $\randset(G,T,p)$. Hence we defer the proof of Lemma~\ref{lem:correctalg} to
Section~\ref{sec:analysis-algorithm}.

\medskip
\textbf{Derandomization.}  In Section~\ref{subsection-derandomization}, we present a deterministic variant of
$\randset(G,T,p)$,  which, instead of returning a random set $Z$, returns a deterministic set $Z_1$, $\dots$, $Z_t$ of
$O^*(2^{2^{O(p)}})$ sets. Instead of bounding the probability that the random set $Z$ has the required property with some
probability, we prove that at least one $Z_i$ always satisfy the property.  Therefore, in Steps 1 and 3 of
Algorithm~\ref{alg:algo}, we can replace $\randset$ with this deterministic variant, and branch on the choice of one $Z_i$
from the returned sets. By the properties of the deterministic algorithm, if $I_1$ is a yes-instance, then $Z$ has Property
(*) in at least one of the branches and therefore the algorithm finds a correct solution for $I_1$. The branching increases
the running time only by a factor of $(O^*(2^{2^{O(p)}}))^2$ and therefore the total running time is $O^*(2^{2^{O(p)}})$.

\section{Important separators and random sampling}
\label{subsection-imp-separators} This section reviews the notion of important separators and the random sampling technique
introduced by Marx and Razgon~\cite{marx-stoc-2011}. As \cite{marx-stoc-2011} used these concepts for undirected graphs and we
need them for directed graphs, we give a self-contained presentation without relying on earlier work.

\subsection{Important separators}
Marx~\cite{marx-2006} introduced the concept of \emph{important separators} to deal with the \textsc{Undirected Multiway Cut}
problem. Since then it has been used implicitly or explicitly in, e.g.,
\cite{chen-improved-multiway-cut,directed-feedback-vertex-set,lokshtanov-marx-clustering,marx-stoc-2011,almost-2-sat} in the
design of fixed-parameter algorithms. In this section, we define and use this concept in the setting of directed graphs.
Roughly speaking, an important separator is a separator of small size that is \emph{maximal} with respect to the set of
vertices on one side.

\begin{definition}{\bf (important separator)}
\label{defn-imp-sep} Let $G$ be a directed graph and let $X,Y\subseteq V$ be two disjoint non-empty sets.  A minimal $X-Y$
separator $S$ is called an \emph{important $X-Y$ separator} if there is no $X-Y$ separator $S'$ with $|S'|\leq |S|$ and
$R^{+}_{G\setminus S}(X)\subset R^{+}_{G\setminus S'}(X)$, where $R^{+}_{A}(X)$ is the set of vertices reachable from $X$ in
$A$.
\end{definition}

Let $X, Y$ be disjoint sets of vertices of an \emph{undirected graph}. Then for every $p\geq 0$ it is
known~\cite{chen-improved-multiway-cut,marx-2006} that there are at most $4^{p}$ important $X-Y$ separators of size at most
$p$ for any sets $X,Y$. The next lemma shows that the same bound holds for important separators even in directed graphs.

\begin{lemma}
{\bf (number of important separators)}
\label{number-of-imp-sep} Let $X,Y\subseteq V(G)$ be disjoint sets in a \emph{directed} graph $G$. Then for every
$p\geq 0$ there are at most $4^p$ important $X-Y$ separators of size at most $p$. Furthermore, we can enumerate all these
separators in time $O(4^p\cdot p(|V(G)+|E(G)|))$.
\end{lemma}
The proof of Lemma~\ref{number-of-imp-sep} is long and follows the
same techniques as the proof in undirected graphs (see e.g.,
\cite{marx-stoc-2011,lokshtanov-marx-clustering}). Therefore, it is
deferred to Appendix~\ref{subsection-number-of-imp-sep} to maintain
the flow of the main result. For ease of notation, we now define the
following collection of important separators:
\begin{definition}\label{defn-forward-reverse-impsep} Given an instance $(G,T,p)$ of
\dmway, the set $\IS$ contains the set $S\subseteq V(G)$ if $S$ is an important $v-T$ separator of size at most $p$ in $G$ for some
vertex $v$ in $V(G)\setminus T$.
\end{definition}

\begin{remark}
\label{remark:number-of-impseps} \emph{It follows from Lemma~\ref{number-of-imp-sep} that $|\IS|\le 4^p\cdot |V(G)|$ and we can enumerate the sets in $\IS$ in time $O^*(4^p)$.}
\end{remark}
We now define a special type of shadows which we use later for the random sampling:

\begin{definition}{\bf (exact shadows)}
\label{defn-exact-shadow} Let $G$ be a directed graph and $T\subseteq V(G)$ a set of terminals. Let $S\subseteq V(G)\setminus
V^{\infty}(G)$ be a set of vertices. Then for $v\in V(G)$ we say that
\begin{enumerate}
\item $v$ is in the ``exact reverse shadow''  of $S$ (with respect to $T$), if $S$ is a minimal $v-T$ separator in $G$,
    and
\item $v$ is in the ``exact forward shadow''  of $S$ (with respect to $T$), if $S$ is a minimal $T-v$ separator in $G$.
\end{enumerate}
\end{definition}

\begin{figure}[t]
\centering
\includegraphics[width=6in]{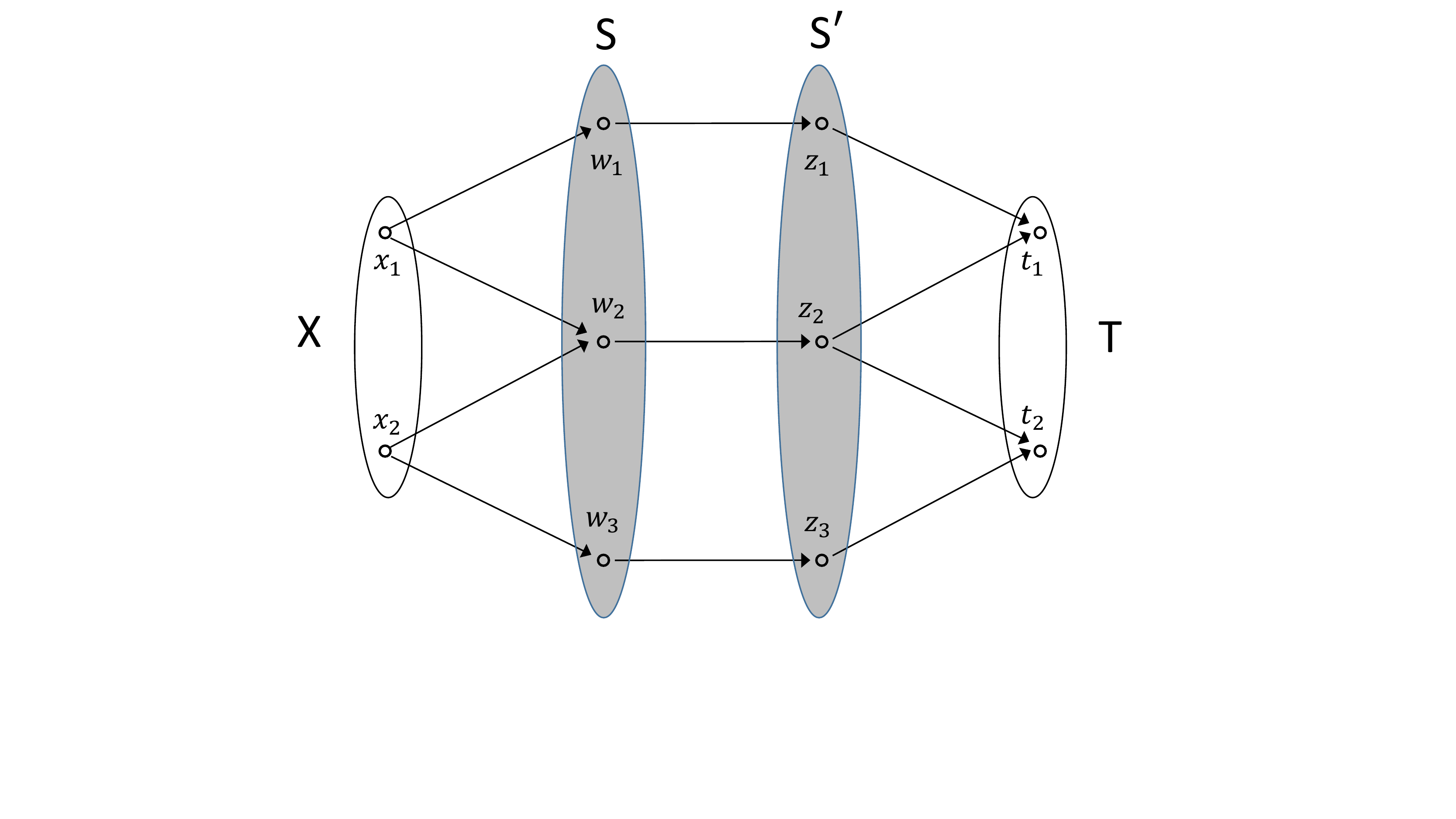}
\vspace{-20mm}
\caption{$S$ is a minimal $X-Y$ separator but it is not an important $X-T$ separator as $S'$ satisfies $|S'|=|S|$ and
 $R^{+}_{G\setminus S}(X) = X \subset X\cup S = R^{+}_{G\setminus S'}(X)$. In fact it is easy to check that the only important $X-T$ separator of size 3 is $S'$.
 If $p\geq 2$ then the set $\{z_1, z_2\}$ is in $\IS$ since it is an important $x_{1}-T$ separator of size $2$. Finally, $x_1$ belongs to the ``exact reverse shadow"
 of each of the sets $\{w_1,w_2\}, \{w_1,z_2\}, \{w_2, z_1\}$ and $\{z_1,z_2\}$ since they are all minimal $x_{1}-T$ separators. However $x_1$ does not belong to the exact reverse shadow of the set $S$ as it is not a minimal $x_{1}-T$ separator.\label{fig:important-separators}}
\end{figure}

We refer the reader to Figure~\ref{fig:important-separators} for
examples of Definitions~\ref{defn-imp-sep},
~\ref{defn-forward-reverse-impsep} and \ref{defn-exact-shadow}. The
exact reverse shadow of $S$ is a subset of the reverse shadow of $S$:
it contains a vertex $v$ only if every vertex $w\in S$ is ``useful''
in separating $v$: vertex $w$ can be reached from $v$ and $T$ can be
reached from $w$. This slight difference between the shadow and the
exact shadow will be crucial in the analysis of the algorithm (see
Section~\ref{sec:analysis-algorithm} and
Remark~\ref{remark:imp-of-exact-shadows}).

The random sampling described in Section~\ref{sec:random-sampling-1} (Theorem~\ref{th:random-sampling}) randomly selects
a members of $\IS$ and creates a subset of vertices by taking the union of the exact reverse shadows of the selected separators.  The following lemma will be
used to give an upper bound on the probability that a vertex is covered by the union.
\begin{lemma}
\label{lemma-bounding-number-of-shadows-per-vertex} Let $z$ be any vertex. Then there are at most $4^{p}$ members of $\IS$ which
contain $z$ in their exact reverse shadows.
\end{lemma}

For the proof of Lemma~\ref{lemma-bounding-number-of-shadows-per-vertex}, we need to establish first the following:

\begin{lemma}
\label{lemma-forward-impsep-exact-back-shadow} If $S\in \IS$ and $v$ is in the exact reverse shadow of $S$, then $S$ is
an important $v-T$ separator.
\end{lemma}
\begin{proof}
  Let $w$ be the witness that $S$ is in $\IS$, i.e., $S$ is an
  important $w-T$ separator in $G$. Let $v$ be any vertex in the exact
  reverse shadow of $S$, which means that $S$ is a minimal $v-T$
  separator in $G$. Suppose that $S$ is not an important $v-T$
  separator. Then there exists a $v-T$ separator $S'$ such that
  $|S'|\leq |S|$ and $R^{+}_{G\setminus S}(v)\subset R^{+}_{G\setminus
    S'}(v)$. We will arrive to a contradiction by showing that
  $R^{+}_{G\setminus S}(w)\subset R^{+}_{G\setminus S'}(w)$, i.e., $S$
  is not an important $w-T$ separator.

First, we claim that $S'$ is an $(S\setminus S')-T$ separator. Suppose that there is a path $P$ from some $x\in S\setminus S'$
to $T$ that is disjoint from $S'$. As $S$ is a minimal $v-T$ separator, there is a path $Q$ from $v$ to $x$ whose internal
vertices are disjoint from $S$. Furthermore, $R^{+}_{G\setminus S}(v)\subset R^{+}_{G\setminus S'}(v)$ implies that  the
internal vertices of $Q$ are disjoint from $S'$ as well. Therefore, concatenating $Q$ and $P$ gives a path from $v$ to $T$
that is disjoint from $S'$, contradicting the fact that $S'$ is a $v-T$ separator.

We show that $S'$ is a $w-T$ separator and its existence contradicts the assumption that $S$ is an important $w-T$ separator.
First we show that $S'$ is a $w-T$ separator. Suppose that there is a $w-T$ path $P$ disjoint from $S'$. Path $P$ has to go
through a vertex $y\in S\setminus S'$ (as $S$ is a $w-T$ separator). Thus by the previous claim, the subpath of $P$ from $y$
to $T$ has to contain a vertex of $S'$, a contradiction.

Finally, we show that $R^{+}_{G\setminus S}(w)\subseteq R^{+}_{G\setminus S'}(w)$. As $S\neq S'$ and $|S'|\le |S|$, this will
contradict the assumption that $S$ is an important $w-T$ separator.
 Suppose that there is a vertex $z \in R^{+}_{G\setminus
  S}(w)\setminus R^{+}_{G\setminus S'}(w)$ and consider a  $w-z$
path that is fully contained in $R^{+}_{G\setminus S}(v)$, i.e., disjoint from $S$.  As $z\not \in R^{+}_{G\setminus S'}(v)$,
path $Q$ contains a vertex $q\in S'\setminus S$.  Since $S'$ is a minimal $v-T$ separator, there is a $v-T$ path that
intersects $S'$ only in $q$. Let $P$ be the subpath of this path from $q$ to $T$. If $P$ contains a vertex $r\in S$, then the
subpath of $P$ from $r$ to $T$ contains no vertex of $S'$ (as $z\neq r$ is the only vertex of $S'$ on $P$), contradicting our
earlier claim that $S'$ is a $(S\setminus S')-T$ separator. Thus $P$ is disjoint from $S$, and hence the concatenation of the
subpath of $Q$ from $w$ to $q$ and the path $P$ is a $w-T$ path disjoint from $S$, a contradiction.
\end{proof}

Lemma~\ref{lemma-bounding-number-of-shadows-per-vertex} easily follows from
Lemma~\ref{lemma-forward-impsep-exact-back-shadow}.  Let $J$ be a member of $\IS$ such that $z$ is in the exact reverse shadow of $J$.
By Lemma~\ref{lemma-forward-impsep-exact-back-shadow}, $J$ is an important $z-T$ separator. By Lemma~\ref{number-of-imp-sep},
there are at most $4^{p}$ important $z-T$ separators of size at most $p$ and so $z$ belongs to at most $4^{p}$ exact reverse
shadows.


\begin{remark}\label{remark:imp-of-exact-shadows}
\emph{It is crucial to distinguish between ``reverse shadow'' and ``exact reverse shadow'':
Lemma~\ref{lemma-forward-impsep-exact-back-shadow} (and hence Lemma~\ref{lemma-bounding-number-of-shadows-per-vertex}) does
not remain true if we remove the word ``exact.'' Consider the following example (see Figure~\ref{fig:imp-of-exact-new}). Let
$a_1$, $\dots$, $a_r$ be vertices such that there is an edge going from every $a_i$ to every vertex of $T=\{t_1, t_2, \ldots,
t_k\}$. For every $1\le i \le r$, let $b_i$ be a vertex with an edge going from $b_i$ to $a_i$. For every $1\le i < j \le r$,
let $c_{i,j}$ be a vertex with two edges going from $c_{i,j}$ to $a_i$ and $a_j$. Then every set $\{a_i,a_j\}$ is in $\IS$,
since it is an important $c_{i,j}-T$ separator. This means that every $b_i$ is in the reverse shadow of $r-1$ members of $\IS$, namely the
sets $\{a_j, ai_i\}$ for $1\leq i\neq j\leq r$. However, $b_i$ is in the {\em exact} reverse shadow of exactly one member of $\IS$, the set $\{a_i\}$.}
\end{remark}

\begin{figure}[t]
\centering
\includegraphics[width=7in]{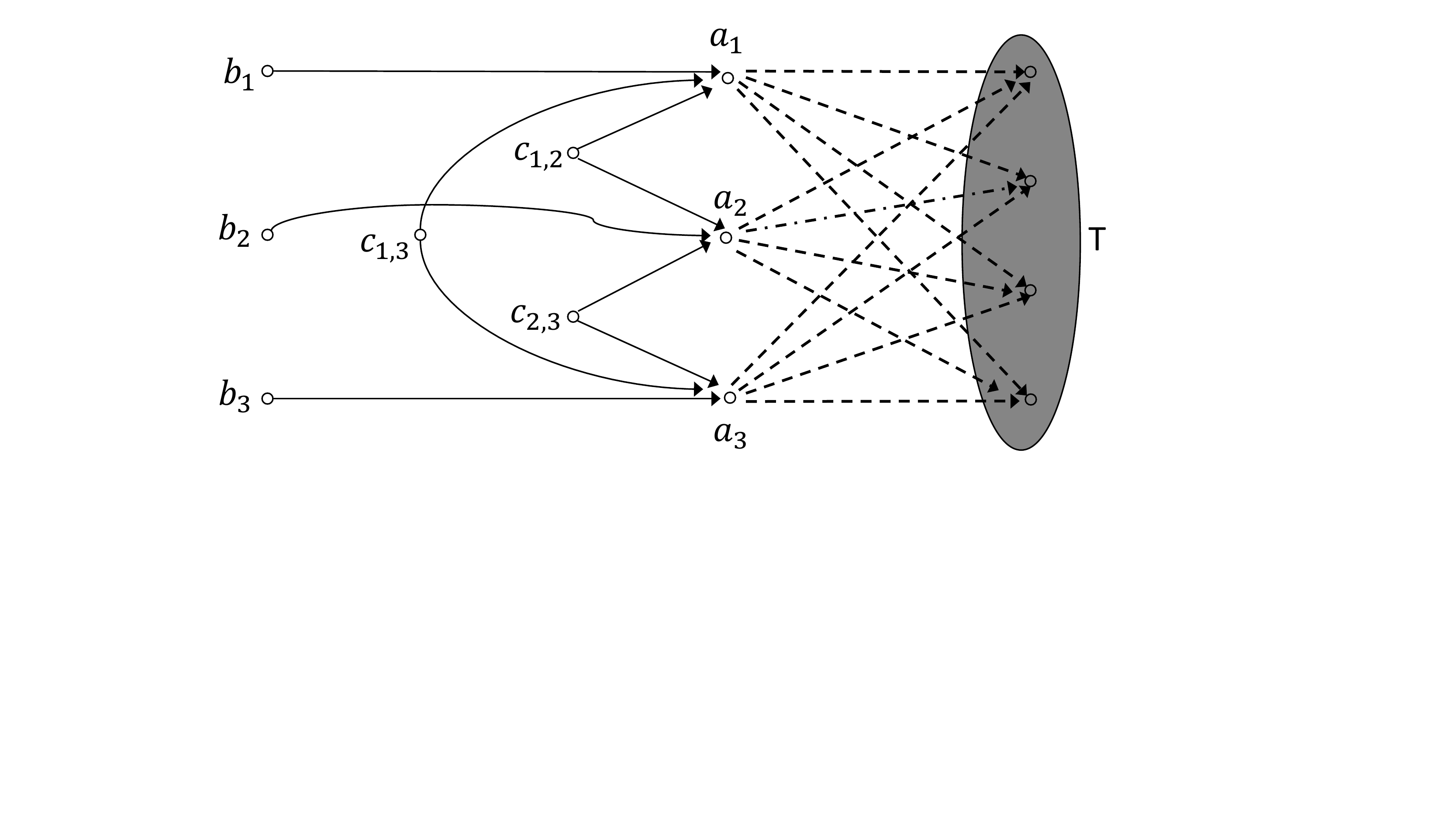}
\vspace{-50mm}
\caption{An illustration of Remark~\ref{remark:imp-of-exact-shadows} in the special case when $k=4$ and $r=3$. \label{fig:imp-of-exact-new}}
\end{figure}

%

\subsection{Random sampling}
\label{sec:random-sampling-1}

In this section, we adapt the random sampling of \cite{marx-stoc-2011} to directed graphs. We try to present it in a
self-contained way that might be useful for future applications.

Roughly speaking, we want to select a random set $Z$ such that for every pair $(S,Y)$ where $Y$ is in the reverse shadow of $S$,
the probability that $Z$ is disjoint from $S$ but contains $Y$ can be bounded from below. We can guarantee such a lower bound
only if $(S,Y)$ satisfies two conditions. First, it is not enough that $Y$ is in the shadow of $S$ (or in other words, $S$ is
an $Y-T$ separator), but $S$ should contain important separators separating the vertices of $Y$ from $T$ (see
Theorem~\ref{th:random-sampling} for the exact statement). Second, a vertex of $S$ cannot be in the reverse shadow of other
vertices of $S$, this is expressed by the following technical definition:
\begin{definition}{\bf (thin)}
\label{defn-thin-set} Let $G$ be a directed graph and $T\subseteq V(G)$ a set of terminals. We say that a set $S\subseteq
V(G)$ is \emph{thin} in $G$ if there is no $v\in S$ such that $v$ belongs to the reverse shadow of $S\setminus v$ with respect
to $T$.
\end{definition}

Refer to Figure~\ref{fig:important-separators}. The sets $S$ is thin because for every $1\leq i\leq 3$ the vertex $w_{i}$ does
not belong to the reverse shadow of the set $S\setminus \{w_i\}$. However the set $S\cup S'$ is not thin since $(S\cup
S')\setminus \{w_1\}$ is a $w_{1}-T$ separator, and hence $w_1$ belongs to the reverse shadow of $(S\cup S')\setminus
\{w_1\}$.

\begin{theorem}{\bf (random sampling)}\label{th:random-sampling}
  There is an algorithm $\randset(G,T,p)$ that produces a random set
  $Z\subseteq V(G)\setminus T$ in time $O^*(4^p)$ such that the
  following holds. Let $S$ be a \emph{thin} set with $|S|\le p$, and let $Y$ be a set such that for every $v\in Y$ there is an important $v-T$ separator $S'\subseteq S$.
  For every such pair $(S,Y)$, the probability that the following two events both occur is at least $2^{-2^{O(p)}}$:
\begin{enumerate}
\item $S\cap Z=\emptyset$, and
\item $Y\subseteq Z$.
\end{enumerate}
\end{theorem}
\begin{proof}
We claim that Algorithm~\ref{alg:sampling} for $\randset(G,T,p)$  satisfies the requirements.
\begin{algorithm}
\caption{$\randset(G,T,p)$} \label{alg:randset}

\begin{algorithmic}[1]\label{alg:sampling}

    \STATE Enumerate every member of $\IS$. \hspace{64mm}\COMMENT{See Remark~\ref{remark:number-of-impseps}}
        \STATE Let $\X$ be the set of exact reverse shadows of members of $\IS$. 
        \STATE Take a random $\mathcal{X}'\subseteq X$ by choosing each element with probability $\frac{1}{2}$, independently
                at random.
        \STATE Let $Z$ be the union of the exact reverse shadows in $\mathcal{X}'$.
        \STATE \textbf{return} $Z$
\end{algorithmic}
\end{algorithm}
The algorithm $\randset(G,T,p)$ first enumerates the collection $\IS$; let $\X$ be the set of all exact reverse
shadows of these sets. By Remark~\ref{remark:number-of-impseps}, the size of $\X$ is $O^*(4^p)$ and it can be constructed
in time $O^*(4^p)$.
Now we show that the set $Z$ satisfies the requirement of the theorem.

Fix a pair $(S,Y)$ as in the statement of the theorem.  Let $X_1,X_2,\ldots,X_d\in \X$ be the exact reverse shadows of every
member of $\IS$ that is a subset of $S$. As $|S|\le p$, we have $d\le 2^p$.  By assumption that $S$ is \emph{thin}, we have $X_j\cap
S=\emptyset$ for every $j\in[d]$.  Now consider the following events:
\begin{enumerate}
\item[(E1)]$Z\cap S= \emptyset$
\item[(E2)] $X_{j}\subseteq Z$ for every $j\in [d]$
\end{enumerate}
Note that (E2) implies that $Y\subseteq Z$.  Our goal is to show that both events (E1) and (E2) occur with probability
$2^{-2^{O(p)}}$.

Let $A=\{X_1,X_2,\ldots,X_d\}$ and $B=\{X\in \X\ |\ X\cap S \neq
\emptyset \}$. By Lemma
\ref{lemma-bounding-number-of-shadows-per-vertex}, each vertex of $S$
is contained in the exact reverse shadow of at most $4^{p}$ members of
$\IS$. Thus $|B|\leq |S|\cdot 4^{p}\leq p\cdot 4^{p}$. If no exact
reverse shadow from $B$ is selected, then event (E1) holds. If every
exact reverse shadow from $A$ is selected, then event (E2) holds. Thus
the probability that both (E1) and (E2) occur is bounded from below by
the probability of the event that every element from $A$ is selected
and no element from $B$ is selected. Note that $A$ and $B$ are
disjoint: $A$ contains only sets disjoint from $S$, while $B$ contains
only sets intersecting $S$.  Therefore, the two events are independent
and the probability that both events occur is at least
\[
\Big(\frac{1}{2}\Big)^{2^{p}}\Big(1-\frac{1}{2}\Big)^{p\cdot 4^{p}} = 2^{-2^{O(p)}}
\]
\end{proof}

\subsection{Derandomization}
\label{subsection-derandomization}

We now derandomize the process of choosing exact reverse shadows in Theorem~\ref{th:random-sampling} using the technique of
\emph{splitters}. An $(n,r,r^2)$-splitter is a family of functions from $[n]\rightarrow [r^2]$ such that $\forall \ M\subseteq
[n]$ with $|M|=r$, at least one of the functions in the family is injective on $M$. Naor et al.~\cite{aravind-focs-1995} give
an explicit construction of an $(n,r,r^2)$-splitter of size $O(r^{6}\cdot \log r\cdot \log n)$.

\begin{theorem}{\bf (deterministic sampling)}\label{th:det-sampling}
  There is an algorithm $\randset(G,T,p)$ that produces $t=2^{2^{O(p)}}$ subsets $Z_1$, $\dots$, $Z_t$ of $V(G)\setminus T$ in time $O^*(2^{2^{O(p)}})$ such that the
  following holds. Let $S$ be a \emph{thin} set with $|S|\le p$, and let $Y$ be a set such that for every $v\in Y$ there is an important $v-T$ separator $S'\subseteq S$.
  For every such pair $(S,Y)$, there is at least one $1\le i \le t$ with
\begin{enumerate}
\item $S\cap Z_i=\emptyset$, and
\item $Y\subseteq Z_i$.
\end{enumerate}
\end{theorem}
\begin{proof}
In the proof of Theorem~\ref{th:random-sampling}, a random subset of a universe $\X$ of size $n_0=|\X|\le 4^p\cdot |V(G)|$ is
selected. We argued that for a fixed $S$, there is a collection $A\subseteq \X$ of $a\le 2^p$ sets and a collection
$B\subseteq \X$ of $b \le p\cdot 4^p$ sets such that if every set in $A$ is selected and no set in $B$ is selected, then
events (E1) and (E2) hold. Instead of the selecting a random subset, we construct several subsets such that at least one of
them satisfies both (E1) and (E2). Each subset is defined by a pair $(h,H)$, where $h$ is a function in an
$(n_0,a+b,(a+b)^2)$-splitter family and $H$ is a subset of $[(a+b)^2]$ of size $a$ (there are $\dbinom{(a+b)^2}{a} =
\dbinom{(2^p+p4^{p})^2}{2^p} = 2^{2^{O(p)}}$ such sets $H$). For a particular choice of $h$ and $H$, we select those exact
shadows $S\in \X$ into $\X'$ for which $h(S)\in H$. The size of the splitter family is $O\Big((a+b)^{6}\cdot \log(a+b)\cdot
\log(n_0)\Big) = 2^{O(p)}\cdot \log|V(G)|$ and the number of possibilities for $H$ is $2^{2^{O(p)}}$. Therefore, we construct
$2^{2^{O(p)}}\cdot\log |V(G)|$ subsets of $\X$.

By the definition of the splitter, there is a function $h$ that is injective on $A\cup B$, and there is a subset $H$ such that
$h(L)\in H$ for every set $L$ in $A$ and $h(M)\not\in H$ for every set $M$ in $B$. For such an $h$ and $H$, the selection will
ensure that (E1) and (E2) hold. Thus at least one of the constructed subsets has the required properties, which is what we
wanted to show.
\end{proof}

\section{Proof of Lemma~\ref{lem:correctalg}}
\label{sec:analysis-algorithm} The goal of this section is to complete the proof of correctness of Algorithm~\ref{alg:algo} by
proving Lemma~\ref{lem:correctalg}. Note that Lemma~\ref{lem:correctalg-part-one} was proved in Section~\ref{overview}.

%

To prove Lemma~\ref{lem:correctalg}, we show that if $I$ is a yes-instance, then there exists a solution $S^*$ for $I_1$ that
remains a solution of the undirected $(G^*_3,T,p)$ as well with probability at least $2^{-2^{O(p)}}$.
Suppose that for some solution $S^*$, the following two properties hold:
\begin{enumerate}
\item $Z\cap S^*=\emptyset$ and
\item $r_{G_1,T}(S^*)\bigcup f_{G_1,T}(S^*)\subseteq Z$.
\end{enumerate}
Then Lemma~\ref{reduced-instance-soln}(2) implies that $S^*$ is a shadowless solution of $I/Z=(G_3,T,p)$. It follows by
Lemma~\ref{lemma-redn-to-undirected-case} that $S^*$ is a solution of the undirected instance $(G^*_3,T,p)$ as well. Thus our
goal is to prove the existence of a solution $S^*$ for which we can give a lower bound on the probability that these two
events occur.

For choosing $S^*$, we need the following definition:
\begin{definition}{\bf (shadow-maximal solution)}
\label{defn-terminal-minimal-solutions} Let $(G,T,p)$ be a given instance of \dmway. An inclusion-wise minimal solution $S$ is
called \emph{shadow-maximal} if $r_{G,T}(S)\bigcup f_{G,T}(S) \bigcup S$ is inclusion-wise maximal among all minimal
solutions.
\end{definition}

For the rest of the proof, let us fix $S^*$ to be a shadow-maximal solution of instance $I_1=(G_1,T,p)$ such that
$|r_{G_1,T}(S^*)|$ is maximum possible among all shadow-maximal solutions. We now give a lower bound on the probability that $Z\cap
S^*=\emptyset$ and $r_{G_1,T}(S^*)\bigcup f_{G_1,T}(S^*)\subseteq Z$. More precisely, we give a lower bound on the probability that all
of the following four events occur:
\begin{enumerate}
\item $Z_1\cap S^*=\emptyset$,
\item $r_{G_1,T}(S^*)\subseteq Z_1$,
\item $Z_2\cap S^*=\emptyset$, and
\item $f_{G_1,T}(S^*)\subseteq Z_2$.
\end{enumerate}
That is, the first random selection takes care of the reverse shadow, the second takes care of the forward shadow, and none of
$Z_1$ or $Z_2$ hits $S^*$. Note that it is somewhat counterintuitive that we choose an $S^*$ for which the shadow is large:
intuitively, it seems that the larger the shadow is, the less likely that it is fully covered by $Z$. However, we need this
maximality property in order to give a lower bound on the probability that $Z\cap S^*=\emptyset$.

We want to invoke Theorem~\ref{th:random-sampling} to obtain a lower bound on the probability that $Z_1$ contains $Y=r_{G_1,T}(S^*)$ and
$Z_1\cap S^*=\emptyset$. First, we need to ensure that $S^*$ is a \emph{thin} set, but this follows easily from the fact that
$S^*$ is a minimal solution:
\begin{lemma}
\label{lemma:minimal-shadow-cover-sol} If $S$ is a minimal solution for a \dmway instance $(G,T,p)$, then no $v\in S$ is in
the reverse shadow of some $S'\subseteq S\setminus \{v\}$.
\end{lemma}
\begin{proof}
  We claim that $S\setminus \{v\}$ is also a solution, contradicting
  the minimality of $S$. Suppose that there is a path $P$ from $t_1\in
  T$ to $t_2\in T$, $t_1\neq t_2$ that intersects $S$ only in
  $v$. Consider the subpath of $P$ from $v$ to $t_2$. As $v$ is in
  $r(S')$, the set $S'$ is a $v-T$ separator. Thus $P$ goes through
  $S'\subseteq S\setminus \{v\}$, a contradiction.
\end{proof}

More importantly, if we want to use Theorem~\ref{th:random-sampling} with $Y=r_{G_1,T}(S^*)$, then we have to make sure that
for every vertex $v$ of $r_{G_1,T}(S^*)$, there is an important $v-T$ separator that is a subset of $S^*$. The ``pushing
argument'' of Lemma~\ref{thm-pushing} shows that if this is not true for some $v$, then we can modify the solution in a way
that increases the size of the reverse shadow. The choice of $S^*$ ensures that no such modification is possible, thus $S^*$
contains an important separator for every $v$.

\begin{lemma}
{\bf (pushing)} \label{thm-pushing} Let $S$ be a solution of a \dmway instance $(G,T,p)$. For every $v\in r(S)$, either there
is an $S_v\subseteq S$ which is an important $v-T$ separator, or there is a solution $S'$ such that
\begin{enumerate}
\item $|S'|\le |S|$,
\item $r(S)\subset r(S')$,
\item $(r(S)\bigcup f(S) \bigcup S) \subseteq (r(S')\bigcup f(S') \bigcup S')$.
\end{enumerate}
\end{lemma}
\begin{proof}
  Let $S_0\subseteq S$ be the subset of $S$ reachable from $v$
  without going through any other vertices of $S$. Then $S_0$ is
  clearly a $v-T$ separator. Let $S_v$ be the minimal $v-T$ separator
  contained in $S_0$.  If $S_v$ is an important $v-T$ separator, then
  we are done as $S$ itself contains $S_v$. Otherwise, there exists an
  important $v-T$ separator $S'_v$, i.e., $|S'_v|\leq |S_v|$ and
  $R^{+}_{G\setminus S_v}(v)\subset R^{+}_{G\setminus
    S'_v}(v)$. Now we show that $S' = (S\setminus S_v)\bigcup S'_v$ is a
  solution for the multiway cut instance. Note that $S'_v\subseteq
  S'$ and $|S'|\leq |S|$.

  First we claim that $r(S)\bigcup (S\setminus S')\subseteq r(S')$. Suppose
  that there is a path $P$ from $\beta$ to $T$ in $G\setminus S'$ for
  some $\beta\in r(S)\bigcup(S\setminus S')$. If $\beta\in r(S)$, then
  path $P$ has to go through a vertex $\beta'\in S$. As $\beta'$ is
  not in $S'$, it has to be in $S\setminus S'$. Therefore, by
  replacing $\beta$ with $\beta'$, we can assume in the following that
  $\beta\in S\setminus S'\subseteq S_v\setminus S'_v$.  By minimality of $S_v$, every vertex
  of $S_v\subseteq S_0$ has an incoming edge from some vertex in
  $R^{+}_{G\setminus S}(v)$. This means that there is a vertex $\alpha
  \in R^{+}_{G\setminus S}(v)$ such that $(\alpha,\beta)\in
  E(G)$. Since $R^{+}_{G\setminus S}(v)\subseteq R^{+}_{G\setminus
    S'}(v)$, we have $\alpha \in R^{+}_{G\setminus S'}(v)$, implying
  that there is a $v\rightarrow \alpha$ path in $G\setminus S'$. The
  edge $\alpha \rightarrow \beta$ also survives in $G\setminus S'$ as
  $\alpha \in R^{+}_{G\setminus S'}(v)$ and $\beta \in S_v \setminus
  S'_v$. By assumption, we have a path in $G\setminus S'$ from
  $\beta$ to some $t\in T$.  Concatenating the three paths we obtain a
  $v \rightarrow t$ path in $G\setminus S'$ which contradicts the fact
  that $S'$ contains an (important) $v-T$ separator $S'_v$.  Since
  $S\neq S'$ and $|S|=|S'|$, the set $S_v\setminus S'_v$ is
  non-empty. Thus $r(S)\subset r(S')$ follows from the claim $r(S)\bigcup
  (S\setminus S')\subseteq r(S')$.

  Suppose now that $S'$ is not a solution for the multiway cut
  instance.  Then there is a $t_1 \rightarrow t_2$ path $P$ in
  $G\setminus S'$ for some $t_1,t_2\in T$, $t_1\neq t_2$. As $S$ is a
  solution for the multiway cut instance, $P$ must pass through a
  vertex $\beta \in S \setminus S' \subseteq r(S')$ (by the
  claim in the previous paragraph), a contradiction. Thus $S'$ is also
  a minimum solution.

  Finally, we show that $r(S)\bigcup f(S)\bigcup S\subseteq r(S')\bigcup
  f(S')\bigcup S'$. We know that $r(S)\bigcup (S\setminus S')\subseteq
  r(S')$. Thus it is sufficient to consider a vertex $v\in
  f(S)\setminus r(S)$. Suppose that $v\not \in f(S')$ and $v\not \in
  r(S')$: there are paths $P_1$ and $P_2$ in $G\setminus S'$, going
  from $T$ to $v$ and from $v$ to $T$, respectively. As $v\in f(S)$,
  path $P_1$ intersects $S$, i.e., it goes through a vertex of $S\setminus S'\subseteq r(S')$; let $\beta$ be the last such vertex on $P_1$. Now concatenating the
  subpath of $P_1$ from $\beta$ to $v$ and the path $P_2$ gives a path
  from $\beta\in r(S')$ to $T$ in $G\setminus S'$, a contradiction.
\end{proof}

Note that if $S$ is a shadow-maximal solution, then solution $S'$ in Lemma~\ref{thm-pushing} is also shadow-maximal.
Therefore, by the choice of $S^*$, applying Lemma~\ref{thm-pushing} on $S^*$ cannot produce a shadow-maximal solution $S'$
with $r_{G_1,T}(S^*)\subset r_{G_1,T}(S')$, and hence $S^*$ contains an important $v-T$ separator for every $v\in
r_{G_1,T}(S)$. Thus by Theorem~\ref{th:random-sampling} for $Y=r_{G_1,T}(S^*)$, we get:
\begin{lemma}
\label{lemma:reverse} With probability at least $2^{-2^{O(p)}}$, both $r_{G_1,T}(S^*)\subseteq Z_1$ and $Z_1\cap
S^*=\emptyset$ occur.
\end{lemma}

In the following, we assume that the events in Lemma~\ref{lemma:reverse} occur.  Our next goal is to give a lower bound on the
probability that $Z_2$ contains $f_{G_1,T}(S^*)$. Note that $S^*$ is a solution also of the instance $(G_2,T,p)$: the vertices
in $S^*$ remained finite (as $Z_1\cap S^*=\emptyset$ by the assumptions of Lemma~\ref{lemma:reverse}), and reversing the orientation of the edges
does not change the fact that $S^*$ is a solution. Solution $S^*$ is a shadow-maximal solution also in $(G_2,T,p)$:
Definition~\ref{defn-terminal-minimal-solutions} is insensitive to reversing the orientation of the edges and making some of
the weights infinite can only decrease the set of potential solutions.  Furthermore, the forward shadow of $S^*$ in $G_2$ is
the same as the reverse shadow of $S^*$ in $G_1$, that is, $f_{G_2,T}(S^*)=r_{G_1,T}(S^*)$. Therefore, assuming that the events in
Lemma~\ref{lemma:reverse} occur, every vertex of $f_{G_2,T}(S^*)$ has infinite weight in $G_2$. We show that now it holds that
$S^*$ contains an important $v-T$ separator in $G_2$ for every $v\in r_{G_2,T}(S^*)=f_{G_1,T}(S^*)$:

\begin{lemma}
\label{lemma-no-pushing-in-terminal-minimal-solns} If $S$ is a shadow-maximal solution for a \dmway instance $(G,T,p)$ and
every vertex of $f(S)$ is infinite, then $S$ contains an important $v-T$ separator for every $v\in r(S)$.
\end{lemma}
\begin{proof}
  Suppose to the contrary that there exists $v\in r(S)$ such that $S$
  does not contain an important $v-T$ separator. Then by
  Lemma~\ref{thm-pushing}, there is a another shadow-maximal
  solution $S'$. As $S$ is shadow-maximal, it follows that $r(S)\bigcup
  f(S)\bigcup S= r(S')\bigcup f(S')\bigcup S'$. Therefore, the nonempty set
  $S'\setminus S$ is fully contained in $r(S)\bigcup f(S)\bigcup
  S$. However it cannot contain any vertex of $f(S)$ (as they are
  infinite by assumption) and cannot contain any vertex of $r(S)$ (as
  $r(S)\subset r(S')$), a contradiction.
\end{proof}

Recall that $S^*$ is a shadow-maximal solution also in $(G_2,T,p)$. In particular, $S^*$ is a minimal solution for $G_2$ and
so by Lemma~\ref{lemma:minimal-shadow-cover-sol} we have that $S^*$ is thin in $G_2$ also. Thus
Theorem~\ref{th:random-sampling} can be used (with $Y=r_{G_2,T}(S^*)$) to obtain a lower bound on the probability that
$r_{G_2,T}(S^*)\subseteq Z_2$ and $Z_2\cap S^*=\emptyset$. As the reverse shadow $r_{G_2,T}(S^*)$ in $G_2$ is the same as the
forward shadow $f_{G_1,T}(S^*)$ in $G_1$, we can state the following:
\begin{lemma}
\label{lemma:forward} Assuming the events in Lemma~\ref{lemma:reverse} occur, with probability at least $2^{-2^{O(p)}}$ both
$f_{G_1,T}(S^*)\subseteq Z_2$ and $Z_2\cap S^*=\emptyset$ occur.
\end{lemma}

Therefore, with probability at least $(2^{-2^{O(p)}})^2$, the set $Z_1\bigcup Z_2$ contains $f_{G_1,T}(S^*)\bigcup
r_{G_1,T}(S^*)$ and it is disjoint from $S^*$. Lemma~\ref{reduced-instance-soln}(2) implies that $S^*$ is a shadowless
solution of $I/(Z_1\bigcup Z_2)$. It follows from Lemma~\ref{lemma-redn-to-undirected-case} that $S^*$ is a solution of the
undirected instance $(G^*_3,T,p)$.
\begin{lemma}
With probability at least $2^{-2^{O(p)}}$, $S^*$ is a shadowless solution of $(G_3,T,p)$ and a solution of the undirected
instance $(G^*_3,T,p)$.
\end{lemma}

In summary, with probability at least $2^{-2^{O(p)}}$ Algorithm~\ref{alg:algo} returns a set $S$ which is a solution of $I$ by
Lemma~\ref{lem:correctalg-part-one}. This completes the proof of Lemma~\ref{lem:correctalg}.


\bibliographystyle{abbrv}
\bibliography{docsdb}

\begin{thebibliography}{10}

\bibitem{DBLP:journals/corr/abs-1010-5197}
N.~Bousquet, J.~Daligault, and S.~Thomass{\'e}.
\newblock Multicut is \textsc{FPT}.
\newblock In {\em STOC}, pages 459--468, 2011.

\bibitem{chen-improved-multiway-cut}
J.~Chen, Y.~Liu, and S.~Lu.
\newblock An improved parameterized algorithm for the minimum node multiway cut
  problem.
\newblock {\em Algorithmica}, 55(1):1--13, 2009.

\bibitem{directed-feedback-vertex-set}
J.~Chen, Y.~Liu, S.~Lu, B.~O'Sullivan, and I.~Razgon.
\newblock A fixed-parameter algorithm for the directed feedback vertex set
  problem.
\newblock {\em J. ACM}, 55(5), 2008.

\bibitem{DBLP:conf/icalp/ChitnisCHM12}
R.~H. Chitnis, M.~Cygan, M.~T. Hajiaghayi, and D.~Marx.
\newblock Directed subset feedback vertex set is fixed-parameter tractable.
\newblock In {\em ICALP (1)}, pages 230--241, 2012.

\bibitem{DBLP:conf/soda/ChitnisHM12}
R.~H. Chitnis, M.~Hajiaghayi, and D.~Marx.
\newblock Fixed-parameter tractability of directed multiway cut parameterized
  by the size of the cutset.
\newblock In {\em SODA}, pages 1713--1725, 2012.

\bibitem{cygan-no-poly-kernel}
M.~Cygan, S.~Kratsch, M.~Pilipczuk, M.~Pilipczuk, and M.~Wahlstr{\"o}m.
\newblock Clique cover and graph separation: New incompressibility results.
\newblock In {\em ICALP (1)}, pages 254--265, 2012.

\bibitem{cygan-ipec-11}
M.~Cygan, M.~Pilipczuk, M.~Pilipczuk, and J.~O. Wojtaszczyk.
\newblock On multiway cut parameterized above lower bounds.
\newblock In {\em IPEC}, pages 1--12, 2011.

\bibitem{johnson}
E.~Dahlhaus, D.~S. Johnson, C.~H. Papadimitriou, P.~D. Seymour, and
  M.~Yannakakis.
\newblock The complexity of multiterminal cuts.
\newblock {\em SIAM J. Comput.}, 23(4):864--894, 1994.

\bibitem{downey-fellows}
R.~G. Downey and M.~R. Fellows.
\newblock {\em Parameterized Complexity}.
\newblock Springer-Verlag, 1999.
\newblock 530 pp.

\bibitem{flum-grohe}
J.~Flum and M.~Grohe.
\newblock {\em Parameterized Complexity Theory}.
\newblock Springer-Verlag, 2006.
\newblock 493 pp.

\bibitem{ford-fulkerson}
L.~Ford and D.~Fulkerson.
\newblock Maximal flow through a network.
\newblock {\em Canad. J. Math.}, 8:399--404, 1956.

\bibitem{MR0159700}
L.~R. Ford, Jr. and D.~R. Fulkerson.
\newblock {\em Flows in networks}.
\newblock Princeton University Press, Princeton, N.J., 1962.

\bibitem{garg}
N.~Garg, V.~V. Vazirani, and M.~Yannakakis.
\newblock Multiway cuts in node weighted graphs.
\newblock {\em J. Algorithms}, 50(1):49--61, 2004.

\bibitem{DBLP:journals/disopt/Guillemot11a}
S.~Guillemot.
\newblock \textsc{FPT} algorithms for path-transversal and cycle-transversal
  problems.
\newblock {\em Discrete Optimization}, 8(1):61--71, 2011.

\bibitem{karger}
D.~R. Karger, P.~N. Klein, C.~Stein, M.~Thorup, and N.~E. Young.
\newblock Rounding algorithms for a geometric embedding of minimum multiway
  cut.
\newblock {\em Math. Oper. Res.}, 29(3):436--461, 2004.

\bibitem{multicut-dags}
S.~Kratsch, M.~Pilipczuk, M.~Pilipczuk, and M.~Wahlstr{\"o}m.
\newblock Fixed-parameter tractability of multicut in directed acyclic graphs.
\newblock In {\em ICALP (1)}, pages 581--593, 2012.

\bibitem{lokshtanov-marx-clustering}
D.~Lokshtanov and D.~Marx.
\newblock Clustering with local restrictions.
\newblock {\em Inf. Comput.}, 222:278--292, 2013.

\bibitem{marx-2006}
D.~Marx.
\newblock Parameterized graph separation problems.
\newblock {\em Theor. Comput. Sci.}, 351(3):394--406, 2006.

\bibitem{marx-stoc-2011}
D.~Marx and I.~Razgon.
\newblock Fixed-parameter tractability of multicut parameterized by the size of
  the cutset.
\newblock In {\em STOC}, pages 469--478, 2011.

\bibitem{aravind-focs-1995}
J.~Naor, L.~Schulman, and A.~Srinivasan.
\newblock Splitters and near-optimal derandomization.
\newblock In {\em FOCS}, 1995.
\newblock pages 182-191.

\bibitem{naor}
J.~Naor and L.~Zosin.
\newblock A 2-approximation algorithm for the directed multiway cut problem.
\newblock In {\em FOCS}, 1997.
\newblock pages 548-553.

\bibitem{niedermeier}
R.~Niedermeier.
\newblock {\em Invitation to Fixed-Parameter Algorithms}.
\newblock Oxford University Press, 2006.
\newblock 312 pp.

\bibitem{almost-2-sat}
I.~Razgon and B.~O'Sullivan.
\newblock Almost 2-\textsc{SAT} is fixed-parameter tractable.
\newblock {\em J. Comput. Syst. Sci.}, 75(8):435--450, 2009.

\bibitem{DBLP:journals/jct/RobertsonS95b}
N.~Robertson and P.~D. Seymour.
\newblock Graph minors {XIII.} {T}he disjoint paths problem.
\newblock {\em J. Comb. Theory, Ser. B}, 63(1):65--110, 1995.

\bibitem{DBLP:journals/jct/RobertsonS04}
N.~Robertson and P.~D. Seymour.
\newblock Graph minors. {XX.} {W}agner's conjecture.
\newblock {\em J. Comb. Theory, Ser. B}, 92(2):325--357, 2004.

\bibitem{DBLP:journals/mst/Xiao10}
M.~Xiao.
\newblock Simple and improved parameterized algorithms for multiterminal cuts.
\newblock {\em Theory Comput. Syst.}, 46(4):723--736, 2010.

\end{thebibliography}


\appendix

\appendix
\section{Bound on the number of important separators (Proof of Lemma~\ref{number-of-imp-sep})}
\label{subsection-number-of-imp-sep}

For the proof of Lemma~\ref{number-of-imp-sep}, we need to establish first some simple properties of important separators, which will allow us to use recursion.
\begin{lemma}
\label{properties-imp-sep} Let $G$ be a directed graph and $S$ be an important $X-Y$ separator. Then
\begin{enumerate}
    \item For every $v\in S$, the set $S\setminus v$ is an important $X-Y$ separator in the graph  $G\setminus v$.
    \item If $S$ is an $X'-Y$ separator for some $X'\supset X$, then $S$ is also an important $X'-Y$ separator.
\end{enumerate}
\end{lemma}
\begin{proof}
~\\
\vspace{-6mm}
\begin{enumerate}
\item Suppose $S\setminus v$ is not a minimal $X-Y$ separator in
  $G\setminus v$. Let $S_{0}\subset S\setminus v$ be an $X-Y$
  separator in $G\setminus v$. Then $S_{0}\cup v$ is an $X-Y$
  separator in $G$, but $S_{0}\cup v \subset S$ holds, which
  contradicts the fact that $S$ is a minimal $X-Y$ separator in
  $G$. Now suppose that there exists an $S'\subseteq V(G)\setminus v$
  such that $|S'|\leq |S\setminus v|=|S|-1$ and $R_{(G\setminus
    v)\setminus (S\setminus v)}^{+}(X) \subset R_{(G\setminus
    v)\setminus S'}^{+}(X)$. Noting that $(G\setminus v)\setminus
  (S\setminus v)=G\setminus S$ and $(G\setminus v)\setminus S'=
  G\setminus (S'\cup v)$, we get $R_{G\setminus S}^{+}(X) \subset
  R_{G\setminus (S'\cup v)}^{+}(X)$. As $|S'\cup v| = |S'|+1 \leq
  |S|$, this contradicts the fact that $S$ is an important $X-Y$
  separator.
\item As $S$ is an inclusionwise minimal $X-Y$ separator, it is an
  inclusionwise minimal $X'-Y$ separator as well. Let $S'$ be a
  witness that $S$ is not an important $X'-Y$ separator in $G$, i.e.,
  $S'$ is an $X'-Y$ separator such that $|S'|\leq |S|$ and
  $R^{+}_{G\setminus S}(X')\subset R^{+}_{G\setminus S'}(X')$. We
  claim first that $R^{+}_{G\setminus S}(X)\subseteq R^{+}_{G\setminus
    S'}(X)$. Indeed, if $P$ is any path from $X$ and fully contained in $R^+_{G\setminus S}(X)$,  then $P$ is disjoint from $S'$, otherwise vertices of
  $P\cap S'$ are in $R^{+}_{G\setminus S}(X')$, but not in
  $R^{+}_{G\setminus S'}(X')$, a contradiction. Next we show that the
  inclusion $R^{+}_{G\setminus S}(X)\subset R^{+}_{G\setminus S'}(X)$
  is proper, contradicting that $S$ is an important $X-Y$ separator.
  As $|S'|\le |S|$, there is a vertex $v\in S\setminus S'$. Since $S$
  is a minimal $X-Y$ separator, it has an in-neighbor $u\in
  R_{G\setminus S}^{+}(X)\subseteq R^+_{G\setminus S'}(X)$. Now $v\in
  S$ and $v\not\in S'$ imply that $v\in R^+_{G\setminus S'}(X)\setminus
  R^+_{G\setminus S}(X)$, a contradiction.
\end{enumerate}
\end{proof}

Next we show that the size of the out-neighborhood of a vertex set is a
submodular function. Recall that a function $f:2^U\rightarrow
\mathbb{N}\cup\{0\}$ is \emph{submodular} if for all $A,B\subseteq U$
we have $f(A)+f(B)\geq f(A\cup B)+ f(A\cap B)$.



\begin{lemma}{\bf (submodularity)}
\label{lemma-out-nbd-is-submodular} The function $\gamma(A)=|N^{+}(A)|$ is submodular.
\end{lemma}
\begin{proof}
  Let $L=\gamma(A)+\gamma(B)$ and $R=\gamma(A\cup B)+\gamma(A\cap
  B)$. To prove $L\geq R$ we show that for each vertex $x\in V$ its
  contribution to $L$ is at least as much as its contribution to
  $R$. Suppose that the weight of $x$ is $w$ (in our setting, $w$ is
  either 1 or $\infty$, but submodularity holds even if the weights
  are arbitrary). The contribution of $x$ to $L$ or $R$ is either $0$, $w$, or $2w$. We have the following four cases:
\begin{enumerate}
\item $x\notin N^{+}(A)$ and $x\notin N^{+}(B)$.\\
  In this case, $x$ contributes 0 to $L$. It contributes 0 to $R$ as
  well: every vertex in $N^+(A\cap B)$ or in $N^+(A\cup B)$ is either
  in $N^+(A)$ or in $N^+(B)$.
\item $x\in N^{+}(A)$ and $x\notin N^{+}(B)$.\\
  In this case, $x$ contributes $w$ to $L$. To see that $x$ does not
  contribute $2w$ to $R$, suppose that $x\in N^{+}(A\cup B)$
  holds. This implies $x\notin A\cup B$ and therefore $x\in
  N^{+}(A\cap B)$ can be true only if $x\in N^+(A)$ and $x\in
  N^{+}(B)$, which is a contradiction.  Therefore, $x$ contributes
  only $w$ to $R$.
\item $x\notin N^{+}(A)$ and $x\in N^{+}(B)$.\\
Symmetric to the previous case.
\item $x\in N^{+}(A)$ and $x\in N^{+}(B)$\\
      In this case, $x$ contributes $2w$ to $L$, and can anyways contribute at most $2w$ to $R$.
\end{enumerate}
In all four cases the contribution of $x$ to $L$ is always greater than or equal to its contribution to $R$ and hence $L\geq R$, i.e.,
$\gamma$ is submodular.
\end{proof}



Recall that $R_{G\setminus S}^{+}(X)$ is the set of vertices reachable from $X$ in $G\setminus S$. The following claim will be useful for the use of submodularity:
\begin{lemma}
\label{lemma-family-of-x-y-separators} Let $G$ be a directed graph. If $S_1, S_2$ are $X-Y$ separators, then both
the sets $N^{+}(R^{+}_{G\setminus S_1}(X)\bigcup R^{+}_{G\setminus S_2}(X))$ and $N^{+}(R^{+}_{G\setminus S_1}(X)\bigcap
R^{+}_{G\setminus S_2}(X))$ are also $X-Y$ separators.
\end{lemma}
\begin{proof}
  1. Let $R_{\cap}=R^{+}_{G\setminus S_1}(X)\bigcap
  R^{+}_{G\setminus S_2}(X)$ and $S_{\cap}=N^{+}(R_\cap)$. As $S_1$ and $S_2$ are disjoint from $X$ and $Y$ by
  definition, we have $X\subseteq R_{\cap}$ and $Y$ is disjoint from
  $R_\cap$. Therefore, every path $P$ from $X$ to $Y$ has a vertex
  $u\in R_{\cap}$ followed by a vertex $v\not\in R_{\cap}$, and
  therefore $v\in S_\cap$. As this holds for every path $P$, the
set $S_\cap$ is an $X-Y$ separator.

2. The argument is the same with the sets $R_{\cup}=R^{+}_{G\setminus
  S_1}(X)\bigcup R^{+}_{G\setminus S_2}(X)$ and
$S_{\cup}=N^{+}(R_\cup)$.
\end{proof}

Now we prove the well-known fact that there is a unique minimum size
separator whose ``reach'' is inclusion-wise maximal.

\begin{lemma}
\label{lemma-unique-min-size-sep-with-maximal-reach} There is a unique $X-Y$ separator $S^{*}$ of minimum size such that
$R_{G\setminus {S^{*}}}^{+}(X)$ is inclusion-wise maximal.
\end{lemma}
\begin{proof}
Let $\lambda$ be the size of a smallest $X-Y$ separator. Suppose to the contrary that there are two separators $S_1$ and $S_2$
of size $\lambda$ such that $R_{G\setminus S_1}^{+}(X)$ and $R_{G\setminus S_2}^{+}(X)$ are incomparable and inclusion-wise
maximal. Let $R_1=R_{G\setminus S_1}^+(X)$, $R_2=R_{G\setminus S_2}^+(X)$, $R_\cap=R_1\cap R_2$, and $R_\cup=R_1\cup R_2$.
 By Lemma~\ref{lemma-out-nbd-is-submodular}, $\gamma$ is submodular and hence
\begin{equation}
\gamma(R_1) + \gamma(R_2) \geq \gamma(R_\cup) + \gamma(R_\cap).
\label{eqn:1}
\end{equation}
As $N^+(R_1)\subseteq S_1$ and $N^+(R_2)\subseteq S_2$, the left hand
side is at most $2\lambda$ (in fact, as $S_1$ and $S_2$ are minimal
$X-Y$ separators, it can be seen that the left hand side is exactly
$2\lambda$).  By Lemma~\ref{lemma-family-of-x-y-separators}, both the
sets $N^{+}(R_\cap)$ and $N^{+}(R_\cup)$ are $X-Y$
separators. Therefore, the right hand side is at least $2\lambda$.
This implies that equality holds in Equation~\ref{eqn:1} and in
particular $|N^{+}(R_\cup)|=\lambda$, i.e., $N^{+}(R_\cup)$ is also a
minimum $X-Y$ separator. As $R_1,R_2\subseteq R_\cup$, every vertex of
$R_1$ and every vertex of $R_2$ is reachable from $X$ in $G\setminus
N^+(R_\cup)$. This contradicts the inclusion-wise maximality of the
reach of $S_1$ and $S_2$.
%
%
\end{proof}

Let $S^*$ be the unique $X-Y$ separator of minimum size given by Lemma~\ref{lemma-unique-min-size-sep-with-maximal-reach}. The
following lemma shows that every important $X-Y$ separator $S$ is ``behind'' this separator $S^{*}$:

\begin{lemma}
\label{lemma-reach-of-S*-contained-in-all-reaches} Let $S^*$ be the unique $X-Y$ separator of minimum size given by
Lemma~\ref{lemma-unique-min-size-sep-with-maximal-reach}. For every important $X-Y$ separator $S$, we have $R_{G\setminus
S^{*}}^{+}(X) \subseteq R_{G\setminus S}^{+}(X)$.
\end{lemma}
\begin{proof}
Note that the condition trivially holds for $S=S^*$. Lemma~\ref{lemma-unique-min-size-sep-with-maximal-reach} implies that the only
important $X-Y$ separator of minimum size is $S^*$.

Suppose there is an important $X-Y$ separator $S\neq S^*$ such that $R_{G\setminus S^{*}}^{+}(X) \nsubseteq R_{G\setminus
S}^{+}(X)$.
 Let $R=R_{G\setminus S}^+(X)$, $R^*=R_{G\setminus S^*}^+(X)$, $R_\cap=R\cap R^*$, and $R_\cup=R\cup R^*$.
 By Lemma~\ref{lemma-out-nbd-is-submodular}, $\gamma$ is submodular and hence
\begin{equation}
\gamma(R^*)+\gamma(R)\geq \gamma(R_{\cup}) + \gamma(R_\cap).
\label{eqn:2}
\end{equation}
As $N^+(R^*)\subseteq S^*$, we have that the first
term on the left hand side is at most $|S^*|=\lambda$.  By
Lemma~\ref{lemma-family-of-x-y-separators}, the set
$N^{+}(R_\cap)$ is an $X-Y$ separator, hence the second term on the
right hand side is at least $\lambda$. It follows that
$|N^{+}(R_\cup)|\leq |N^{+}(R)(X))|\le |S|$.  Since
$R^* \nsubseteq R$ by assumption, we
have $R\subset R_\cup$.  By
Lemma~\ref{lemma-family-of-x-y-separators}, $N^+(R_\cup)$ is also an $X-Y$
separator and we have seen that it has size at most
$|S|$. Furthermore, $R\subset R_\cup$ implies that
any vertex reachable from $X$ in $G\setminus S$ is reachable in
$G\setminus N^{+}(R_\cup)$ as well, contradicting the assumption that $S$
is an important separator.
%
\end{proof}

Now we finally have all the required tools to prove Lemma~\ref{number-of-imp-sep}.

\begin{proof}[Proof (of Lemma \ref{number-of-imp-sep})]
  Let $\lambda$ be the size of a smallest $X-Y$ separator. To prove
  Lemma~\ref{number-of-imp-sep}, we show by induction on $2p-\lambda$
  that the number of important $X-Y$ separators of size at most $p$ is
  upper bounded by $2^{2p-\lambda}$.  Note that if $2p-\lambda < 0$,
  then $\lambda > 2p\geq p$ and so there is no (important) $X-Y$
  separator of size at most $p$. If $2p-\lambda = 0$, then $\lambda =
  2p$. Now if $p=0$ then $\lambda=p=0$ and the empty set is the unique
  important $X-Y$ separator of size at most $p$. If $p>0$, then
  $\lambda = 2p>p$ and hence there is no important $X-Y$ separator of
  size at most $p$. Thus we have checked the base case for
  induction. From now on, the induction hypothesis states that if
  $X',Y'\subseteq V(G)$ are disjoint sets such that $\lambda'$ is the
  size of a smallest $X'-Y'$ separator and $p'$ is an integer such
  that $(2p'-\lambda')<(2p-\lambda)$, then the number of important
  $X'-Y'$ separators of size at most $p'$ is upper bounded by
  $2^{2p'-\lambda'}$.

Let $S^*$ be the unique $X-Y$ separator of minimum size given by Lemma~\ref{lemma-unique-min-size-sep-with-maximal-reach}.
Consider an arbitrary vertex $v\in S^{*}$. Note that $\lambda
> 0$ and so $S^{*}$ is not empty. Any important $X-Y$ separator $S$ of size at most $p$ either contains $v$ or not. If $S$
contains $v$, then by Lemma \ref{properties-imp-sep}(1), the set $S\setminus \{v\}$ is an important $X-Y$ separator in
$G\setminus v$ of size at most $p':= p-1$. As $v\notin X\cup Y\cup V^{\infty}$, the size $\lambda'$ of the minimum $X-Y$
separator in $G\setminus v$ is at least $\lambda-1$. Therefore, $2p'-\lambda' = 2(p-1)-\lambda' = 2p-(\lambda'+2)< 2p-\lambda$.
The induction hypothesis implies that there are at most $2^{2p'-\lambda'}\leq 2^{2p-\lambda-1}$ important $X-Y$ separators of
size $p'$ in $G\setminus v$. Hence there are at most $2^{2p-\lambda-1}$ important $X-Y$ separators of size at most $p$ in $G$
that contain $v$.

Now we give an upper bound on the number of important $X-Y$ separators
{\em not} containing $v$. By minimality of $S^{*}$, vertex $v$ has an
in-neighbor in $R_{G\setminus S^{*}}^{+}(X)$. For every important
$X-Y$ separator $S$,
Lemma~\ref{lemma-reach-of-S*-contained-in-all-reaches} implies
$R_{G\setminus S^{*}}^{+}(X) \subseteq R_{G\setminus S}^{+}(X)$. As
$v\notin S$ and $v$ has an in-neighbor in $R_{G\setminus
  S^{*}}^{+}(X)$, even $R_{G\setminus S^{*}}^{+}(X)\bigcup \{v\}
\subseteq R_{G\setminus S}^{+}(X)$ holds. Therefore, setting
$X'=R_{G\setminus S^{*}}^{+}(X)\bigcup \{v\}$, the set $S$ is also an
$X'-Y$ separator.  Now Lemma \ref{properties-imp-sep}(2) implies that
$S$ is in fact an important $X'-Y$ separator. Since $S$ is an $X-Y$
separator, we have $|S|\geq \lambda$. We claim that in fact
$|S|>\lambda$: otherwise $|S|=|S^*|=\lambda$ and $R_{G\setminus
  S^{*}}^{+}(X)\bigcup \{v\} \subseteq R_{G\setminus S}^{+}(X)$,
contradicting the fact that $S^*$ is an important $X-Y$ separator. So
the minimum size $\lambda'$ of an $X'-Y$ separator in $G$ is at least
$\lambda+1$. By the induction hypothesis, the number of important
$X'-Y$ separators of size at most $p$ in $G$ is at most
$2^{2p-\lambda'}\leq 2^{2p-\lambda-1}$. Hence there are at most
$2^{2p-\lambda-1}$ important $X-Y$ separators of size at most $p$ in
$G$ that do not contain $v$.

Adding the bounds in the two cases, we get the required upper bound of
$2^{2p-\lambda}$.  An algorithm for enumerating all the at most $4^p$
important separators follows from the above proof. First, we can find
a maximum $X-Y$ flow in time $O(p(|V(G)|+|E(G)|))$ using at most $p$
rounds of the Ford-Fulkerson algorithm, where $n$ and $m$ are the
number of vertices and edges of $G$. It is well-known that the
separator $S^*$ of
Lemma~\ref{lemma-unique-min-size-sep-with-maximal-reach} can be
deduced from the maximum flow in linear time by finding those vertices
from which $Y$ cannot be reached in the residual graph
\cite{MR0159700}.  Pick any arbitrary vertex $v\in S^*$. Then we
branch on whether vertex $v\in S^{*}$ is in the important separator or
not, and recursively find all possible important separators for both
cases.  The formal description is given in
Algorithm~\ref{alg:imp-separators}.  Note that this algorithm
enumerates a superset of all important separators: by our analysis
above, every important separator appears in either $S'_1$ or $S_2$,
but there is no guarantee that all the separators in these sets are
important. Therefore, the algorithm has to be followed by a filtering
phase where we check for each returned separator whether it is
important. Observe that $S$ is an important $X-Y$ separator if and
only if $S$ is the unique minimum $R^+_{G\setminus S}(X)-Y$
separator. As the size of $S$ is at most $p$, this can be checked in
time $O(p(|V(G)|+|E(G)|))$ by finding a maximum flow and constructing
the residual graph.  The search tree has at most $4^p$ leaves and the
work to be done in each node is $O(p(|V(G)|+|E(G)|))$. Therefore, the
total running time of the branching algorithms is $O(4^p\cdot
p(|V(G)|+|E(G)|))$ and returns at most $4^p$ separators. This is followed by the
filtering phase, which takes time $O(4^p\cdot p(|V(G)|+|E(G)|))$.
\end{proof}

\begin{algorithm}
\caption{\textsc{ImpSep$(G,X,Y,p)$}} \label{alg:imp-separators}
~\\
\textbf{Input}: A directed graph $G$, disjoint sets $X, Y\subseteq V$ and an integer $p$.\\
\textbf{Output}: A collection of $X-Y$ separators that is a superset of all important $X-Y$ separators of size at most $p$ in $G$.\\

\begin{algorithmic}[1]

    \STATE Find the minimum $X-Y$ separator $S^*$ of Lemma~\ref{lemma-unique-min-size-sep-with-maximal-reach}\hspace{37mm}
    \STATE Let $\lambda=|S'|$
    \IF {$p<\lambda$}
                 \STATE \textbf{return} $\emptyset$
        \ELSE
                \STATE Pick any arbitrary vertex $v\in S^*$
                \STATE Let $\mathcal{S}_1=$\textsc{ImpSep$(G\setminus \{v\},X,Y,p-1)$}
                \STATE Let $\mathcal{S}'_1 = \{v\cup S\ |\ S\in \mathcal{S}_1\}$
                \STATE Let $X'=R_{G\setminus S^{*}}^{+}(X)\cup \{v\}$
                \STATE Let $\mathcal{S}_2=$\textsc{ImpSep$(G,X',Y,p)$}
                        \STATE \textbf{return} $\mathcal{S}'_1 \cup \mathcal{S}_2$
    \ENDIF
\end{algorithmic}
\end{algorithm}

\end{document}